\documentclass{article}

\usepackage{times}
\usepackage{graphicx} 
\usepackage{subfigure} 

\usepackage{natbib}

\usepackage{algorithm}
\usepackage{algorithmic}

\usepackage{hyperref}
\usepackage{xcolor}
\definecolor{Ckt}{HTML}{E41A1C}
\definecolor{Min}{HTML}{4D4D4D}
\definecolor{MinMore}{HTML}{377EB8}
\definecolor{Data}{HTML}{984EA3}

\usepackage{fullpage}

\usepackage{tikz}
\usetikzlibrary{arrows}
\usepackage{amssymb,amsmath}
\usepackage{natbib}
\usepackage{amsthm}

\newtheorem{theorem}{Theorem}
\newtheorem{definition}{Definition}
\DeclareMathOperator*{\argmin}{arg\,min}

\usepackage{stfloats}

\DeclareMathOperator*{\minimize}{minimize}
\newcommand{\ZZ}{\mathbb Z}

\newcommand{\RR}{\mathbb R}

\begin{document}

\title{A log-linear time algorithm for constrained changepoint detection}
\author{
Toby Dylan Hocking (toby.hocking@r-project.org)\\
Guillem Rigaill (guillem.rigaill@inra.fr)\\
Paul Fearnhead (p.fearnhead@lancaster.ac.uk)\\
Guillaume Bourque (guil.bourque@mcgill.ca)
}
\maketitle

\begin{abstract}
  Changepoint detection is a central problem in time series and
  genomic data. For some applications, it is natural to impose
  constraints on the directions of changes. One example is ChIP-seq
  data, for which adding an up-down constraint improves peak detection
  accuracy, but makes the optimization problem more complicated. We
  show how a recently proposed functional pruning technique can be
  adapted to solve such constrained changepoint detection
  problems. This leads to a new algorithm which can solve problems
  with arbitrary affine constraints on adjacent segment means, and
  which has empirical time complexity that is log-linear in the amount
  of data. This algorithm achieves state-of-the-art accuracy in a
  benchmark of several genomic data sets, and is orders of magnitude
  faster than existing algorithms that have similar accuracy. Our
  implementation is available as the PeakSegPDPA function in the coseg
  R package, \url{https://github.com/tdhock/coseg}
\end{abstract}

\newpage

\tableofcontents

\newpage

\section{Introduction}

Changepoint detection is a central problem in fields such as finance
or genomics, where $n$ data are gathered in a sequence over time or
space. Many models define the optimal changepoints using maximum
likelihood, resulting in a discrete optimization problem. Multiple
changepoint detection models seek the optimal $K$ segments
($K-1$ changes), which amounts to optimizing likelihood
parameters over a space that contains $O(n^{K-1})$ discrete
arrangements of changepoints. In general this problem can be solved
in $O(Kn^2)$ time using the original dynamic programming algorithm of
\citet{segment-neighborhood}. Recently proposed pruning techniques
reduce the number of changepoints considered by the algorithm, thus
reducing time complexity to $O(K n\log n)$ while maintaining
optimality \citep{pruned-dp, johnson, fpop}.

In ``unconstrained'' changepoint models, there are no contraints
between model parameters on separate segments. To regularize and obtain a more
interpretable model, it is often desirable to introduce constraints
between model parameters before and after changepoints. 
For example, the main problem that motivates this paper is peak
detection in ChIP-seq data, which provide noisy measurements of
protein binding or modification throughout a genome \citep{practical}. An
up-down constrained changepoint detection model has been shown to
achieve state-of-the-art peak detection accuracy in ChIP-seq data
\citep{HOCKING-PeakSeg}. The constraints of this model force an up
change in the segment mean parameter after each down change, and vice
versa.
The fastest existing solver for this problem is the Constrained
Dynamic Programming Algorithm (CDPA), which has two issues. First, it
is a heuristic algorithm that is not guaranteed to recover the optimal
solution. Second, its $O(Kn^2)$ quadratic time complexity is too slow
for use on large data sets. In this
paper we propose a new algorithm that fixes both of these issues.

\subsection{Contributions and organization}

We begin by discussing previous research into pruning techniques for
solving unconstrained changepoint detection problems
(Section~\ref{sec:related}), then state the constrained optimization
problems (Section~\ref{sec:models}). Our main contribution is
Section~\ref{sec:algorithms}, which generalizes the functional pruning
technique of \citet{pruned-dp}, thus providing a new Generalized
Pruned Dynamic Progamming Algorithm (GPDPA) for solving a class of
constrained changepoint detection problems. We show that the GPDPA
achieves state-of-the-art speed and accuracy in genomic data with
several different labeled patterns (Section~\ref{sec:results}), then
conclude by discussing the significance of our contributions
(Section~\ref{sec:discussion}).

\begin{table*}[b!]
  \centering
  \begin{tabular}{r|c|c}
    & No pruning & Functional pruning \\
    \hline
    Unconstrained & Dynamic Programming Algorithm (DPA) & Pruned DPA (PDPA) \\
    & Optimal solution, $O(Kn^2)$ time & Optimal solution, $O(Kn\log n)$ time\\
    & \citet{segment-neighborhood}     & \citet{pruned-dp, johnson} \\
    \hline
    Up-down constrained & Constrained DPA (CDPA) & Generalized Pruned DPA (GPDPA) \\
    & Sub-optimal solution, $O(Kn^2)$ time & Optimal solution, $O(Kn\log n)$ time\\
    & \citet{HOCKING-PeakSeg} & \textbf{This paper} \\
    \hline
  \end{tabular}
  \caption{Our contribution is 
the Generalized Pruned Dynamic Programming Algorithm (GPDPA), 
 which uses a functional pruning technique 
    to compute the constrained optimal $K-1$ changepoints 
in a sequence of $n$ data, in $O(K n\log n)$ time on average.}
\label{tab:contribution}
\end{table*}

\section{Related work}
\label{sec:related}

There are many efficient algorithms available for computing the
optimal $K-1$ changepoints in $n$ data
points. \citet{segment-neighborhood} proposed an $O(K n^2)$ algorithm
for computing the sequence of models with $1,\dots,K$ segments.
\citet{optimal-partitioning} consider a related approach, which
introduces a penalty for each changepoint, rather than fixing the
number of changepoints. Their $O(n^2)$ algorithm computes the single
model for a given penalty constant $\lambda$. Both of these algorithms
recover the optimal solution, and follow from using dynamic
programming updates \citep{bellman} to recursively compute the maximum
likelihood from 1 to $n$ data points. Alternatively there are methods
which are computationally faster but are not guaranteed to find the
optimal segmentation. The most popular of these is the binary
segmentation algorithm which has $O(Kn)$ worst-case time complexity
\citep{binary-segmentation}. An L1 relaxation of this problem is known
as the fused lasso signal approximator, for which efficient solvers
also exist \citep{flsa}.

Several pruning methods have been recently proposed in order to reduce
time complexity, while maintaining optimality.  \citet{pruned-dp} and
\citet{phd-johnson} independently discovered a functional pruning
technique, which results in algorithms with $O(n\log n)$ average time
complexity. \citet{pelt} proposed an inequality pruning technique,
which results in an algorithm with average time complexity from $O(n)$
to $O(n^2)$, depending on the number of changes. \citet{fpop} provides
a clear discussion on the differences between the two pruning
techniques.

All algorithms discussed thus far are for solving problems with no
constraints between adjacent segment mean parameters, but there are
many examples of constrained changepoint detection models. Rather
than searching all possible changepoints and likelihood parameters,
the idea is to use a constraint in order to search a smaller, more
interpretable model space. For example, \citet{haiminen2008algorithms}
propose an $O(Kn^2)$ algorithm for unimodal regression, which enforces
no up changes after the first down change. \citet{HOCKING-PeakSeg}
proposed an $O(Kn^2)$ algorithm for peak detection, which enforces a
down change after each up change, and vice versa.

Isotonic regression is another example of a constrained changepoint
detection model. There is no limit on the number of segments $K$, but
the segment means are constrained to be non-decreasing. This problem
can be solved in $O(n)$ time using the pool-adjacent-violators
algorithm \citep{mair2009isotone}, or in $O(n\log n)$ time using a
dynamic programming algorithm \citep{isotonic-dp}. An L1 relaxation of
this problem is known as nearly-isotonic regression
\citep{tibshirani2011nearly}. A problem known as reduced isotonic
regression occurs by imposing an additional constraint of $K$ segments
\citep{reduced-monotonic-regression}. The techniques for solving this
problem lead to sub-quadratic time algorithms
\citep{hardwick2014optimal}, but do not generalize to other kinds of
constraints (such as unimodal regression or peak detection).

Our contribution in this paper is proving that the functional pruning
technique can be generalized to constrained changepoint models
(Table~\ref{tab:contribution}). Our resulting Generalized Pruned
Dynamic Programming Algorithm (GPDPA) enjoys $O(Kn\log n)$ time
complexity, and works for any changepoint model with affine
constraints between adjacent segment means (including isotonic
regression, unimodal regression, and peak detection).


\section{Isotonic regression and changepoint models}
\label{sec:models}

Although our proposed algorithm can solve many constrained changepoint
detection problems (Section~\ref{sec:general}), we will simplify our
discussion by emphasizing the isotonic regression model. 

\subsection{Classical isotonic regression}

The classical isotonic regression model is defined as the most likely
sequence of non-decreasing segment means. More
precisely, assume that the data $\mathbf y\in\RR^n$ are a realization
of a probability distribution with mean parameter $\mathbf m\in\RR^n$. For
example, assuming $y_t \sim \mathcal N(m_t, \sigma^2)$ and performing
maximum likelihood inference results in a convex minimization problem
with affine constraints,
\begin{align}
  \label{eq:isotonic}
  \minimize_{\mathbf m\in\RR^n} &\ \ 
  \sum_{t=1}^n \ell(y_t, m_t)\\
  \text{subject to} & \ \ m_t \leq m_{t+1},\, \forall t<n.
  \nonumber
\end{align}
The convex loss function $\ell:\RR\times \RR\rightarrow\RR$ in the
case of the Gaussian likelihood is the square loss
$\ell(y, m) = (y-m)^2$. This optimization problem (\ref{eq:isotonic})
is referred to as isotonic regression, and can be efficiently solved
in $O(n)$ time using the Pool-Adjacent-Violators Algorithm (PAVA)
\citep{isotonic-unifying}.

Since isotonic regression imposes no limit on the number of
changepoints ($m_t < m_{t+1}$), it tends to overfit. For example,
consider the toy data set $\mathbf y= \left[
\begin{array}{cccccc}
  2 & 5 & 30 & 34 & 600 & 621
\end{array}
\right] \in\RR^6$. Because these data are strictly increasing, the
isotonic regression (\ref{eq:isotonic}) solution is the trivial model
$m_t=y_t$. However, these data contain only two large changes. To
recover these changes, we could instead use the segment neighborhood
model, which we discuss in the next section.

\subsection{Segment neighborhood changepoint model}

The segment neighborhood model of \citet{segment-neighborhood} uses
the same cost function as isotonic regression, but a different
constraint set. There is no constraint on the direction of changes,
but there must be exactly $K\leq n$ distinct segments ($K-1$ changes).
\begin{align}
  \label{eq:optimal_segment_neighborhood}
  \minimize_{\mathbf m\in\RR^n} &\ \ 
  \sum_{t=1}^n \ell(y_t, m_t)\\
  \text{subject to} &\ \  \sum_{t=1}^{n-1} I(m_t \neq m_{t+1}) = K-1.
  \nonumber
\end{align}
This optimization problem is non-convex since the model complexity is
the number of changepoints, measured via the non-convex indicator
function $I$. Nonetheless, the optimal solution can be computed in
$O(K n^2)$ time using the standard dynamic programming algorithm
\citep{segment-neighborhood}. By exploiting the structure of the
convex loss function $\ell$, the pruned dynamic programming algorithm
of \citet{pruned-dp} computes the same optimal solution in faster
$O(K n \log n)$ time.

Unlike isotonic regression, the segment neighborhood model does not
constrain the direction of the changes. Thus, for some data sets
$\mathbf y$, the segment neighborhood model may recover a change down
($m_t > m_{t+1}$). For applications where isotonic regression is used,
it would be desirable to compute a model with $K$ non-decreasing
segment means. This results in the reduced isotonic regression
problem, which we introduce in the next section.

\subsection{Reduced isotonic regression}

The idea of fitting a non-decreasing function with a limited number of
changepoints has been previously described as reduced isotonic
regression \citep{reduced-monotonic-regression}. Combining the
constraints of the isotonic regression (\ref{eq:isotonic}) and segment
neighborhood (\ref{eq:optimal_segment_neighborhood}) problems gives
\begin{align}
  \label{eq:reduced}
  \minimize_{\mathbf m\in\RR^n} &\ \ 
  \sum_{t=1}^n \ell(y_t, m_t)\\
  \text{subject to} &\ \  \sum_{t=1}^{n-1} I(m_t \neq m_{t+1}) = K-1,
  \nonumber\\
  &\ \  m_t \leq m_{t+1},\, \forall t<n.
  \nonumber 
\end{align}
In the next section, we explain how functional pruning
can be used for solving this and related changepoint problems.

\newcommand{\FCC}{C}
\newcommand{\M}{\mathcal{M}}
\section{Functional 
pruning algorithms for constrained
  changepoint models}
\label{sec:algorithms}


We begin by discussing an algorithm for solving the reduced isotonic
regression problem, then explain how the algorithm
generalizes to other constrained changepoint problems.

\subsection{Equivalent optimization space}

The reduced isotonic regression problem (\ref{eq:reduced}) has $n$
segment mean variables $m_t$, one for each data point $t$. To derive
our algorithm, we re-write the problem in terms of the mean
$u_k\in\RR$ and endpoint $t_k\in\{1,\dots,n\}$ for each
segment $k\in\{1,\dots, K\}$.
\begin{definition}[Reduced isotonic regression optimization space]
\label{def:Ibar}
  Let $(\mathbf u, \mathbf t)\in{\mathcal I}^n_K$ be the set of
  non-decreasing segment means $u_1\leq\cdots\leq u_K$ and
  increasing changepoint indices $0=t_0<t_1<\cdots<t_{K-1}<t_K=n$.
\end{definition}
Each segment mean $u_k$ is assigned to data points
$\tau\in(t_{k-1},t_k]\subset\{1,\dots,n\}$, resulting in the following
cost for each segment $k\in\{1, \dots, K\}$, 
\begin{equation}
  \label{eq:h}
  h_{t_{k-1}, t_k}(u_k) = \sum_{\tau=t_{k-1}+1}^{t_k} \ell(y_\tau, u_k).
\end{equation}
The reduced isotonic regression problem can be equivalently written as
\begin{equation}
  \label{eq:isotonic_ut}
  \minimize_{(\mathbf u, \mathbf t)\in{\mathcal I}^n_K}
  \sum_{k=1}^K
  h_{t_{k-1}, t_k}(u_k)
\end{equation}
Rather than explicitly summing over data points $i$ as in problem
(\ref{eq:reduced}), this problem uses the equivalent sum over segments $k$.

\begin{figure*}[t!]
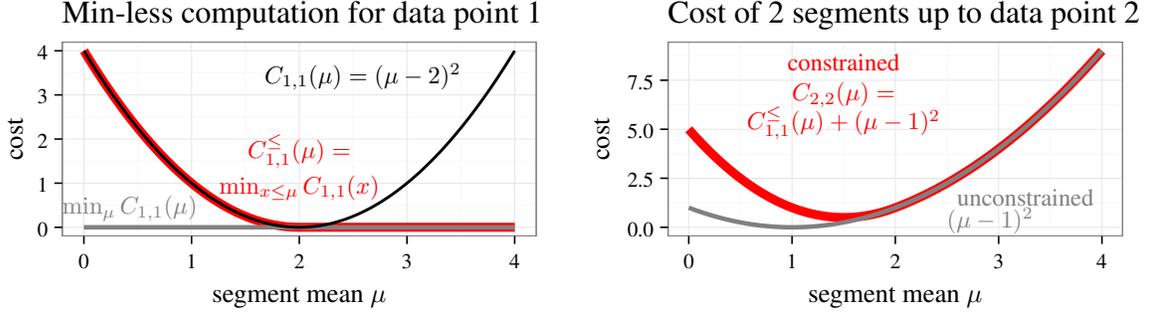

  \centering
  \input{figure-compare-unconstrained}
  \input{figure-compare-cost}
  \vskip -0.5cm
  \caption{Comparison of previous unconstrained algorithm
    (\textcolor{Min}{grey}) with new algorithm that constrains segment
    means to be non-decreasing (\textcolor{Ckt}{red}), for the toy data
    set $\mathbf y= [ 2, 1, 0, 4 ] \in\RR^4$ and the square
    loss. \textbf{Left:} rather than computing the unconstrained
    minimum (constant grey function), the new algorithm computes the
    min-less operator (red), resulting in a larger cost when the
    segment mean is less than the first data point ($\mu <
    2$). \textbf{Right:} adding the cost of the second data point
    $(\mu-1)^2$ and minimizing yields equal means $u_1=u_2=1.5$ for
    the constrained model and decreasing means $u_1=2,\, u_2=1$ for
    the unconstrained model.}
  \label{fig:compare-unconstrained}
\end{figure*}




\subsection{Dynamic programming update rules}
\label{sec:dyn-prog}
Optimization problem (\ref{eq:isotonic_ut}) has $K$ segment mean
variables $u_k$ and $K-1$ changepoint index variables $t_k$. Minimizing over all
variables except the last segment mean $u_K$ results in the following
definition of the optimal cost.



\begin{definition}[Optimal cost with last segment mean $\mu$]
\label{def:fcc}
  Let $\FCC_{K,n}(\mu)$ be the optimal cost of the segmentation
  with $K$ segments, up to data point $n$, with last segment mean
  $\mu$:
\begin{equation}
\FCC_{K,n}(\mu) = \min_{(\mathbf u, \mathbf t)\in{\mathcal I}^n_K \ | \ u_K = \mu} \
  \left\{ \sum_{k=1}^K
  h_{t_{k-1}, t_k}(u_k) \right\}.
\end{equation}
\end{definition}

As in the PDPA of \citet{pruned-dp}, our proposed dynamic programming
algorithm uses an exact representation of the
$C_{k,t}:\RR\rightarrow\RR$ cost functions. Each $C_{k,t}(\mu)$ is
represented as a piecewise function on intervals of $\mu$. This is
implemented as a linked list of FunctionPiece objects in C++ (for
details see Section~\ref{sec:pseudocode}). Each element of the linked list
represents a convex function piece, and implementation details depend
on the choice of the loss function $\ell$ (for an example using the
square loss see Section~\ref{sec:example-comparison}).

In the original unconstrained PDPA, computing the $C_{k,t}(u_k)$
function requires taking the minimum of $C_{k,t-1}(u_k)$ (a function
of the last segment mean $u_k$) and
$\hat C_{k-1,t-1} = \min_{u_{k-1}} C_{k-1,t-1}(u_{k-1})$ (the constant
loss resulting from an unconstrained minimization with respect to the
previous segment mean $u_{k-1}$). The main novelty of our paper is the
discovery that this update can also be computed efficiently for
constrained problems. For example in reduced isotonic regression the second
term is no longer a constant, but instead a function of $u_k$,
$C_{k-1,t-1}^{\leq}(u_k) = \min_{u_{k-1}\leq u_k}
C_{k-1,t-1}(u_{k-1})$, which we refer to as the min-less operator
(Figure~\ref{fig:compare-unconstrained}, left).

\begin{definition}[Min-less operator]
\label{def:min-less}
  Given any real-valued function $f:\RR\rightarrow\RR$, we define the min-less
  operator of that function as $f^\leq(\mu)=\min_{x\leq \mu} f(x)$.
\end{definition}

The min-less operator is used in the following Theorem, which states
the update rules used in our proposed algorithm.

\begin{theorem}[Generalized Pruned Dynamic Programming Algorithm
  for reduced isotonic regression]
\label{thm:gpdpa}
  The optimal cost functions $C_{k,t}$ can be recursively computed
  using the following update rules.
\begin{enumerate}
\item For $k=1$ we have
$\FCC_{1,1}(\mu)=\ell(y_1,\mu)$, and for the other data
  points $t>1$ we have
\begin{equation}
\FCC_{1,t}(\mu)=\FCC_{1,t-1}(\mu)+\ell(y_t,\mu)
\end{equation}
\item For $k>1$ and $t=k$ we have
\begin{equation}
  \FCC_{k,k}(\mu)=\ell(y_k, \mu)+\FCC_{k-1,k-1}^\leq(\mu)
\end{equation}
\item In all other cases we have
  \begin{equation}
  \FCC_{k,t}(\mu)=\ell(y_t,\mu)+
  \min\{
  \FCC_{k-1,t-1}^\leq(\mu),\,
  \FCC_{k,t-1}(\mu)
  \}.
  \end{equation}
\end{enumerate}
\end{theorem}

\begin{proof}
  Case 1 and 2 follow from Definition~\ref{def:fcc}, and there is a proof for case 3 in Section~\ref{sec:proof}.

\end{proof}

The dynamic programming algorithm requires computing $O(Kn)$ cost
functions $\FCC_{k,t}$. As in the original pruned dynamic programming
algorithm, the time complexity of the algorithm is $O(K n I)$ where
$I$ is the number of intervals (convex function pieces; candidate
changepoints) that are used to represent the cost functions. The
theoretical maximum number of intervals is $I=O(n)$, implying a time
complexity of $O(K n^2)$ \citep{pruned-dp-new}. However, this maximum
is only achieved in pathological synthetic data sets, such as a
monotonic increasing data sequence. The average number of intervals in
real data sets is empirically $I=O(\log n)$, as we will show in
Section~\ref{sec:results_time}. Thus the average time complexity of
the algorithm is $O(K n \log n)$.

\subsection{Example and comparison with unconstrained case}
\label{sec:example-comparison}

To clarify the discussion, consider the 
toy data set $\mathbf y= \left[
\begin{array}{cccccc}
  2 & 1 & 0 & 4
\end{array}
\right] \in\RR^4$ and the square loss $\ell(y,\mu)=(y-\mu)^2$. The first
step of the algorithm is to compute the minimum and the maximum of the
data (0,4) in order to bound the possible values of the segment
mean $\mu$. Then the algorithm computes the optimal cost in $k=1$ segment up
to data point $t=1$:
\begin{equation}
  \FCC_{1,1}(\mu) = (2-\mu)^2=4 - 4\mu + \mu^2\text{ (for $\mu\in[0,4]$)}
\end{equation}
This function can be stored for all values of $\mu$ via the three
real-valued coefficients ($\text{constant}=4$, $\text{linear}=-4$,
$\text{quadratic}=1$). To compute the optimal cost in $K=2$ segments,
we first compute the min-less operator (red curve on left of
Figure~\ref{fig:compare-unconstrained}),
\begin{equation}
  \FCC_{1,1}^\leq(\mu) =
  \begin{cases}
    4 - 4\mu + \mu^2 &\text{ if }\mu\in[0,2],\, \mu'=\mu,\\
    0 + 0\mu + 0\mu^2 & \text{ if }\mu\in[2,4],\,  \mu'=2.
  \end{cases}
\end{equation}
This function can be stored as a list of two
intervals of $\mu$ values, each with associated real-valued
coefficients. In addition, to facilitate recovery of the optimal
parameters, we store the previous segment mean $\mu'$ and endpoint
(not shown). Note that $\mu'=\mu$ means that the equality constraint
is active ($u_1=u_2$).

By adding the first min-less function $\FCC_{1,1}^\leq(\mu)$ to the
cost of the second data point $(\mu-2)^2$ we obtain the optimal cost in $K=2$
segments up to data point $t=2$,
\begin{equation}
  \FCC_{2,2}(\mu) = 
  \begin{cases}
    5 - 6\mu + 2\mu^2 &\text{ if }\mu\in[0,2],\,  \mu'=\mu,\\
    1 - 2\mu + 1\mu^2 &\text{ if }\mu\in[2,4],\,  \mu'=2.
  \end{cases}
\end{equation}
Note that the minimum of this function is achieved at $\mu=1.5$ which
occurs in the first of the two function pieces (red curve on right of
Figure~\ref{fig:compare-unconstrained}), with an equality constraint
active. This implies the optimal model up to data point $t=2$ with
$k=2$ non-decreasing segment means actually has no change
($u_1=u_2=1.5$). In contrast, the minimum of the cost computed by the
unconstrained algorithm is at $u_2=1$ (grey curve on right of
Figure~\ref{fig:compare-unconstrained}), resulting in a change down
from $u_1=2$.

\subsection{The PeakSeg up-down constraint}
\label{sec:PeakSeg}

The PeakSeg model described by \citet{HOCKING-PeakSeg} is the most
likely segmentation where the first change is up, all up changes are
followed by down changes, and all down changes are followed by up
changes. More precisely, the constrained optimization problem can be
stated as
\begin{align}
  \label{eq:PeakSeg}
  \minimize_{
        \substack{\mathbf u\in\RR^K \\
    0=t_0<t_1<\cdots<t_{K-1}<t_K=n
}
    } &\ \ 
  \sum_{k=1}^K h_{t_{k-1}, t_k}(u_k)\\
      \text{subject to \hskip 0.9cm} &\ \ u_{k-1} \leq u_k\ \forall k\in\{2,4,\dots\},
  \nonumber\\
  &\ \ u_{k-1} \geq u_k\ \forall k\in\{3,5,\dots\}.
  \nonumber
\nonumber
\end{align}

\begin{figure*}[t!]
  \centering
  \input{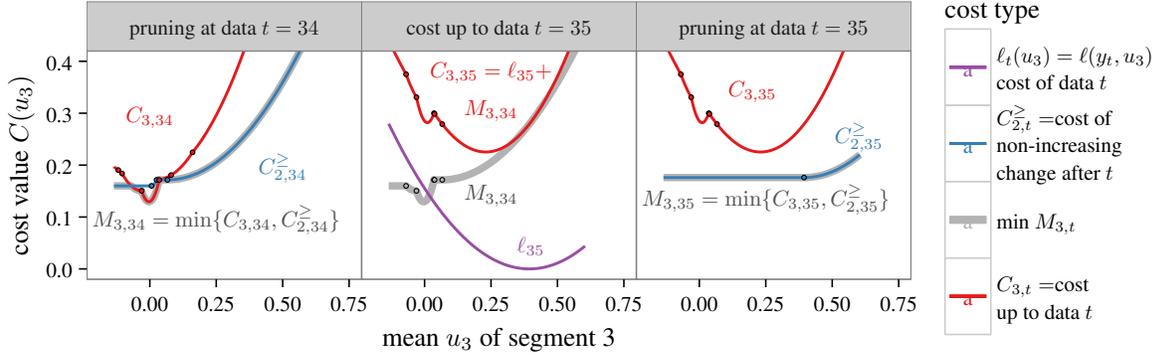}
\vskip -0.5cm
  \caption{
    Demonstration of GPDPA for the PeakSeg model~(\ref{eq:PeakSeg})
    with $k=3$ segments. Cost functions are stored as piecewise
    functions on intervals (black dots show limits between function
    pieces). \textbf{Left:} the min \textcolor{Min}{$M_{3,34}$} is the
    minimum of two functions: \textcolor{MinMore}{$C^{\geq}_{2,34}$}
    is the cost if the second segment ends at data point $t=34$ (the
    min-more operator forces a non-increasing change after), and
    \textcolor{Ckt}{$C_{3,34}$} is the cost if the second segment ends
    before that. \textbf{Middle:} the cost \textcolor{Ckt}{$C_{3,35}$}
    is the sum of the min \textcolor{Min}{$M_{3,34}$} and the cost of
    the next data point \textcolor{Data}{$\ell_{35}$}. \textbf{Right:}
    in the next step, all previously considered changepoints are
    pruned (cost \textcolor{Ckt}{$C_{3,35}$}), since the model with a the second
    segment ending at data point $t=35$ is always less costly
    (\textcolor{MinMore}{$C^{\geq}_{2,35}$}).  }
  \label{fig:min-envelope}
\end{figure*}

Our proposed Generalized Pruned Dynamic Programming Algorithm (GPDPA)
can be used to solve the PeakSeg problem. The initialization $k=1$ is
the same as in the reduced isotonic regression solver
(Section~\ref{sec:dyn-prog}). The dynamic programming updates for even
$k\in\{2, 4, \dots\}$ are also the same. However, to constrain non-increasing
changes, the updates for odd $k\in\{3, 5, \dots\}$ are
\begin{equation}
  \FCC_{k,t}(\mu) = \ell(y_t, \mu) + \min\{
  \FCC_{k-1,t-1}^\geq(\mu),\, \FCC_{k,t-1}(\mu)
  \},
\end{equation}
where the min-more operator is defined for any function $f:\RR\rightarrow\RR$ as
$f^\geq(\mu) = \min_{x\geq \mu} f(x)$. Figure~\ref{fig:min-envelope}
shows the geometric interpretation of the min-more operator, along
with an example of how the $\min\{\}$ operation performs pruning.
We implemented this algorithm using the Poisson loss
$\ell(y, \mu) = \mu - y\log \mu$, since our application in
Section~\ref{sec:results-chip-seq} is on count data
$y\in\ZZ_+ = \{0, 1, 2, \dots\}$.
We implemented this algorithm in C++, and our free/open-source code is
available as the PeakSegPDPA function in the coseg R package for
constrained optimal segmentation
(\url{https://github.com/tdhock/coseg}). Implementation details can be
found in Section~\ref{sec:pseudocode}.

\subsection{General affine inequality constraints
  between adjacent segment means}
\label{sec:general}

In this section we briefly discuss how our proposed Generalized Pruned
Dynamic Programming Algorithm (GPDPA) can be used to solve any
optimization problem with affine inequality constraints
between adjacent segment means. For each change $k\in\{1,\dots,K-1\}$,
let $a_k,b_k,c_k\in\RR$ be arbitrary coefficients that define affine
functions $g_k(u_k, u_{k+1})=a_k u_k + b_k u_{k+1} + c_k$. The
changepoint detection problem with general affine constraints is
\begin{align}
  \label{eq:min_general_affine_inequality}
  \minimize_{
    \substack{
    \mathbf u\in\RR^K\\
0=t_0<t_1<\cdots<t_{K-1}<t_K=n
}
    } &\ \ 
  \sum_{k=1}^K h_{t_{k-1}, t_k}(u_k)\\
  \text{subject to \hskip 0.9cm} &\ \  \forall k\in\{1,\dots,K-1\},
  \nonumber\\
&\ \ g_k(u_k, u_{k+1})\leq 0.\nonumber                        
\nonumber
\end{align}


Some examples of models that are special cases:
\begin{enumerate}
\item If we take all $a_k,b_k,c_k=0$ then the constraints are
  trivially satisfied, we
  recover the unconstrained segment neighborhood problem
  (\ref{eq:optimal_segment_neighborhood}).
\item If we take all $a_{k} =1$, $b_{k}=-1$ and $c_{k} = 0$ we recover
  the reduced isotonic regression problem
  (\ref{eq:isotonic_ut}).
\item For the PeakSeg problem (\ref{eq:PeakSeg}),
  we take all $c_{k} = 0$. For odd $k\in\{1,3,\dots\}$ we take
  $a_{k} =1$, $b_{k}=-1$ and for even $k\in\{2,4,\dots\}$ we take
  $a_{k} =-1$, $b_{k}=1$.
\end{enumerate}
To solve these problems, we need to compute the analog of the
min-less/more operator, which we call the constrained minimization
operator. For any cost function $f:\RR\rightarrow\RR$ and constraint
function $g:\RR\times\RR\rightarrow\RR$, we define the constrained
minimization operator $f^g:\RR\rightarrow\RR$ as
\begin{equation}
  \label{eq:constrained-min-operator}
  f^g(u_{k}) = \min_{u_{k-1} : g(u_{k-1}, u_{k})\leq 0} f(u_{k-1}).
\end{equation}
When $g$ is affine, the constrained minimization operator is either
non-decreasing or non-increasing. In this case it can be computed
using a simple algorithm that scans the piecewise function $f$ either
from left to right or right to left. When a local minimum is found,
its value is recorded, and a constant function piece is added (for
details see pseudocode for MinLess algorithm in
Section~\ref{sec:MinLess}). The constrained minimization operator is
used in the following general dynamic programming update rule which
can be used to compute the solution to
(\ref{eq:min_general_affine_inequality})
\begin{equation}
  \label{eq:general_dp}
  \FCC_{k,t}(\mu) = \ell(y_t,\mu) + \min\{
  \FCC_{k,t-1}(\mu),\,
  \FCC_{k-1,t-1}^{g_{k-1}}(\mu)
  \}.
\end{equation}
We note that this update rule is valid for constraint functions $g$
more general than affine functions. However, the
closed-form computation of the constrained minimization operator
(\ref{eq:constrained-min-operator}) would possibly be much more
difficult for these more general constraint functions (e.g. quadratic
constraint functions).

\section{Results on peak detection in ChIP-seq data}
\label{sec:results-chip-seq}
\label{sec:results}

The real data analysis problem that motivates this work is the
detection of peaks in ChIP-seq data \citep{practical}, which are
typically represented as a vector of non-negative counts
$\mathbf y\in\ZZ_+^n$ of aligned sequence reads for $n$ continguous
bases in a genome. Data sizes are between $n=10^5$ (maximum of the
benchmark we consider) and $n=10^8$ (largest region with no gaps in
the human genome hg19). A peak detector can be represented
as a function $c(\mathbf y)\in\{0,1\}^n$ for binary classification at
every base position. The positive class is peaks (genomic regions with
large values, representing protein binding or modification) and the
negative class is background noise (small values).

In the supervised learning framework of \citet{HOCKING2016-chipseq}, a
data set consists of $m$ count data vectors
$\mathbf y_1,\dots,\mathbf y_m$ along with labels $L_1,\dots, L_m$
that identify regions with and without peaks. Briefly, the number of
errors $E[c(\mathbf y_i), L_i]$ is the total of false positives
(negative labels with a predicted peak) plus false negatives (positive
labels with no predicted peak). The benchmark consists of seven
histone ChIP-seq data sets, each with a different peak pattern
(experiment type, labeler, cell types). The goal in each data set is
to learn the pattern encoded in the labels, and find a classifier $c$
that minimizes the total number of incorrectly predicted labels in a
held-out test set:
\begin{equation}
  \label{eq:learn}
  \minimize_c
  \sum_{i=1}^m 
  E\left[
    c(\mathbf y_i), L_i
  \right].
\end{equation}

\citet{HOCKING-PeakSeg} proposed a constrained dynamic programming
algorithm (CDPA) to approximately compute the optimal changepoints,
subject to the PeakSeg up-down constraint
(Section~\ref{sec:PeakSeg}). The CDPA has been shown to achieve
state-of-the-art peak detection accuracy, by classifying even-numbered
segments $k$ as peaks, and odd-numbered segments $k$ as background
noise. However, its quadratic $O(Kn^2)$ time complexity makes it too
slow to run on large ChIP-seq data sets.

\begin{figure*}[t!]
  \centering
  \parbox{0.49\textwidth}{
\begin{tikzpicture}[x=1pt,y=1pt]
\definecolor{fillColor}{RGB}{255,255,255}
\path[use as bounding box,fill=fillColor,fill opacity=0.00] (0,0) rectangle (238.49,130.09);
\begin{scope}
\path[clip] (  0.00,  0.00) rectangle (238.49,130.09);
\definecolor{drawColor}{RGB}{255,255,255}
\definecolor{fillColor}{RGB}{255,255,255}

\path[draw=drawColor,line width= 0.6pt,line join=round,line cap=round,fill=fillColor] (  0.00,  0.00) rectangle (238.49,130.09);
\end{scope}
\begin{scope}
\path[clip] ( 53.92, 33.48) rectangle (226.49,124.09);
\definecolor{fillColor}{RGB}{255,255,255}

\path[fill=fillColor] ( 53.92, 33.48) rectangle (226.49,124.09);
\definecolor{drawColor}{gray}{0.98}

\path[draw=drawColor,line width= 0.6pt,line join=round] ( 53.92, 76.31) --
	(226.49, 76.31);

\path[draw=drawColor,line width= 0.6pt,line join=round] ( 53.92,110.22) --
	(226.49,110.22);

\path[draw=drawColor,line width= 0.6pt,line join=round] ( 60.51, 33.48) --
	( 60.51,124.09);

\path[draw=drawColor,line width= 0.6pt,line join=round] ( 84.23, 33.48) --
	( 84.23,124.09);

\path[draw=drawColor,line width= 0.6pt,line join=round] (131.67, 33.48) --
	(131.67,124.09);

\path[draw=drawColor,line width= 0.6pt,line join=round] (106.52, 33.48) --
	(106.52,124.09);

\path[draw=drawColor,line width= 0.6pt,line join=round] (140.21, 33.48) --
	(140.21,124.09);
\definecolor{drawColor}{gray}{0.90}

\path[draw=drawColor,line width= 0.2pt,line join=round] ( 53.92, 52.15) --
	(226.49, 52.15);

\path[draw=drawColor,line width= 0.2pt,line join=round] ( 53.92,100.48) --
	(226.49,100.48);

\path[draw=drawColor,line width= 0.2pt,line join=round] ( 53.92,119.97) --
	(226.49,119.97);

\path[draw=drawColor,line width= 0.2pt,line join=round] (107.95, 33.48) --
	(107.95,124.09);

\path[draw=drawColor,line width= 0.2pt,line join=round] (155.39, 33.48) --
	(155.39,124.09);

\path[draw=drawColor,line width= 0.2pt,line join=round] ( 57.64, 33.48) --
	( 57.64,124.09);

\path[draw=drawColor,line width= 0.2pt,line join=round] (222.78, 33.48) --
	(222.78,124.09);
\definecolor{drawColor}{RGB}{0,0,0}

\node[text=drawColor,anchor=base,inner sep=0pt, outer sep=0pt, scale=  0.92] at (107.95,111.18) {max};

\node[text=drawColor,anchor=base,inner sep=0pt, outer sep=0pt, scale=  0.92] at (155.39, 37.62) {median, inter-quartile range};

\path[draw=drawColor,line width= 0.6pt,line join=round] ( 61.77, 80.17) --
	( 70.02, 87.55) --
	( 78.28, 90.11) --
	( 86.54, 83.25) --
	( 94.79, 85.07) --
	(103.05, 96.06) --
	(111.31,102.48) --
	(119.56,108.43) --
	(127.82,111.50) --
	(136.08,112.94) --
	(144.33,115.55) --
	(152.59,119.97) --
	(160.85,113.96) --
	(169.11,112.35) --
	(177.36,116.16) --
	(185.62,109.55) --
	(193.88,100.06) --
	(202.13, 90.11) --
	(210.39, 89.05) --
	(218.65, 88.31);
\definecolor{fillColor}{RGB}{0,0,0}

\path[fill=fillColor,fill opacity=0.50] ( 61.77, 52.15) --
	( 70.02, 55.08) --
	( 78.28, 49.93) --
	( 86.54, 51.07) --
	( 94.79, 54.15) --
	(103.05, 57.65) --
	(111.31, 60.66) --
	(119.56, 62.01) --
	(127.82, 62.01) --
	(136.08, 62.01) --
	(144.33, 64.48) --
	(152.59, 65.62) --
	(160.85, 66.70) --
	(169.11, 64.48) --
	(177.36, 64.48) --
	(185.62, 64.48) --
	(193.88, 64.48) --
	(202.13, 65.62) --
	(210.39, 65.62) --
	(218.65, 65.62) --
	(218.65, 57.65) --
	(210.39, 55.97) --
	(202.13, 55.97) --
	(193.88, 55.97) --
	(185.62, 55.97) --
	(177.36, 54.15) --
	(169.11, 54.15) --
	(160.85, 54.15) --
	(152.59, 54.15) --
	(144.33, 52.15) --
	(136.08, 49.93) --
	(127.82, 49.93) --
	(119.56, 47.46) --
	(111.31, 47.46) --
	(103.05, 44.66) --
	( 94.79, 41.42) --
	( 86.54, 41.42) --
	( 78.28, 37.60) --
	( 70.02, 39.60) --
	( 61.77, 37.60) --
	cycle;

\path[draw=drawColor,line width= 0.6pt,line join=round] ( 61.77, 46.11) --
	( 70.02, 47.46) --
	( 78.28, 44.66) --
	( 86.54, 46.11) --
	( 94.79, 49.93) --
	(103.05, 52.15) --
	(111.31, 54.15) --
	(119.56, 54.15) --
	(127.82, 55.97) --
	(136.08, 55.97) --
	(144.33, 59.21) --
	(152.59, 59.21) --
	(160.85, 60.66) --
	(169.11, 59.21) --
	(177.36, 60.66) --
	(185.62, 60.66) --
	(193.88, 60.66) --
	(202.13, 60.66) --
	(210.39, 60.66) --
	(218.65, 62.01);
\definecolor{drawColor}{gray}{0.50}

\path[draw=drawColor,line width= 0.6pt,line join=round,line cap=round] ( 53.92, 33.48) rectangle (226.49,124.09);
\end{scope}
\begin{scope}
\path[clip] (  0.00,  0.00) rectangle (238.49,130.09);
\definecolor{drawColor}{RGB}{0,0,0}

\node[text=drawColor,anchor=base east,inner sep=0pt, outer sep=0pt, scale=  0.80] at ( 48.52, 48.84) {10};

\node[text=drawColor,anchor=base east,inner sep=0pt, outer sep=0pt, scale=  0.80] at ( 48.52, 97.18) {100};

\node[text=drawColor,anchor=base east,inner sep=0pt, outer sep=0pt, scale=  0.80] at ( 48.52,116.66) {253};
\end{scope}
\begin{scope}
\path[clip] (  0.00,  0.00) rectangle (238.49,130.09);
\definecolor{drawColor}{RGB}{0,0,0}

\path[draw=drawColor,line width= 0.6pt,line join=round] ( 50.92, 52.15) --
	( 53.92, 52.15);

\path[draw=drawColor,line width= 0.6pt,line join=round] ( 50.92,100.48) --
	( 53.92,100.48);

\path[draw=drawColor,line width= 0.6pt,line join=round] ( 50.92,119.97) --
	( 53.92,119.97);
\end{scope}
\begin{scope}
\path[clip] (  0.00,  0.00) rectangle (238.49,130.09);
\definecolor{drawColor}{RGB}{0,0,0}

\path[draw=drawColor,line width= 0.6pt,line join=round] (107.95, 30.48) --
	(107.95, 33.48);

\path[draw=drawColor,line width= 0.6pt,line join=round] (155.39, 30.48) --
	(155.39, 33.48);

\path[draw=drawColor,line width= 0.6pt,line join=round] ( 57.64, 30.48) --
	( 57.64, 33.48);

\path[draw=drawColor,line width= 0.6pt,line join=round] (222.78, 30.48) --
	(222.78, 33.48);
\end{scope}
\begin{scope}
\path[clip] (  0.00,  0.00) rectangle (238.49,130.09);
\definecolor{drawColor}{RGB}{0,0,0}

\node[text=drawColor,anchor=base,inner sep=0pt, outer sep=0pt, scale=  0.80] at (107.95, 21.46) {1000};

\node[text=drawColor,anchor=base,inner sep=0pt, outer sep=0pt, scale=  0.80] at (155.39, 21.46) {10000};

\node[text=drawColor,anchor=base,inner sep=0pt, outer sep=0pt, scale=  0.80] at ( 57.64, 21.46) {87};

\node[text=drawColor,anchor=base,inner sep=0pt, outer sep=0pt, scale=  0.80] at (222.78, 21.46) {263169};
\end{scope}
\begin{scope}
\path[clip] (  0.00,  0.00) rectangle (238.49,130.09);
\definecolor{drawColor}{RGB}{0,0,0}

\node[text=drawColor,anchor=base,inner sep=0pt, outer sep=0pt, scale=  1.00] at (140.21,  8.40) {$n$ = data points to segment (log scale)};
\end{scope}
\begin{scope}
\path[clip] (  0.00,  0.00) rectangle (238.49,130.09);
\definecolor{drawColor}{RGB}{0,0,0}

\node[text=drawColor,rotate= 90.00,anchor=base,inner sep=0pt, outer sep=0pt, scale=  1.00] at ( 16.66, 78.78) {$I$ = intervals stored};

\node[text=drawColor,rotate= 90.00,anchor=base,inner sep=0pt, outer sep=0pt, scale=  1.00] at ( 29.62, 78.78) {(log scale)};
\end{scope}
\end{tikzpicture}
  }
  \parbox{0.49\textwidth}{
    \input{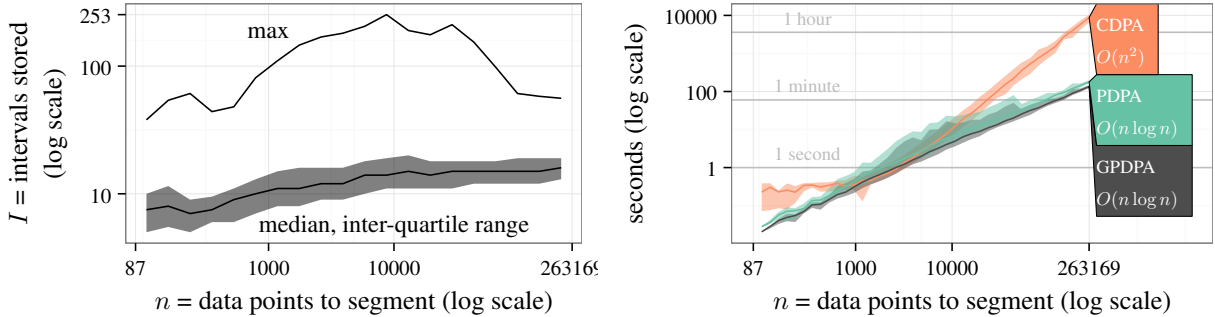}
  }
  \vskip -0.8cm
  \caption{Empirical speed analysis on 2752 count data vectors from
    the histone mark ChIP-seq benchmark. For each vector we ran the
    GPDPA with the up-down constraint and a max of $K=19$
    segments. The expected time complexity is $O(KnI)$ where $I$ is
    the average number of intervals (function pieces; candidate
    changepoints) stored in the $C_{k,t}$ cost
    functions. \textbf{Left}: number of intervals stored is
    $I=O(\log n)$ (median, inter-quartile range, and maximum over all
    data points $t$ and segments $k$).  \textbf{Right}: time
    complexity of the GPDPA is $O(n\log n)$ (median line and min/max
    band).}
  \label{fig:timings}
\end{figure*}

In this section, we show that our proposed GPDPA can be used to
overcome this speed drawback, while maintaining state-of-the-art
accuracy. To show the importance of enforcing the up-down constraint,
we consider the unconstrained Pruned Dynamic Programming Algorithm
(PDPA) of \citet{pruned-dp} as a baseline
(Table~\ref{tab:contribution}). We also compare against two popular
heuristics from the bioinformatics literature, in order to demonstrate
that constrained optimization algorithms such as the CDPA and GPDPA
are more accurate.

\subsection{Empirical time complexity in ChIP-seq data}
\label{sec:results_time}

The ChIP-seq benchmark consists of seven labeled histone data
sets.
Overall there are 2752 count data vectors $\mathbf y_i$ to segment,
varying in size from $n=87$ to $n=263169$ data. For each count data
vector $\mathbf y_i$, we ran each algorithm (CDPA, PDPA, GDPDA) with a
maximum of $K=19$ segments. This implies a maximum of 9 peaks (one for
each even-numbered segment), which is more than enough in these
relatively small data sets. To analyze the empirical time complexity,
we recorded the number of intervals stored in the $\FCC_{k,t}$ cost
functions (Section~\ref{sec:algorithms}), as well as the computation
time in seconds.

As in the PDPA, the time complexity of our proposed GPDPA is
$O(K n I)$, which depends on the number of intervals $I$ (candidate
changepoints) stored in the $\FCC_{k,t}$ cost functions
\citep{pruned-dp-new}. We observed that the number of intervals stored
by the GPDPA increases as a sub-linear function of the number of data
points $n$ (left of Figure~\ref{fig:timings}). For the largest data
set ($n=263169$), the algorithm only stored median=16 and maximum=43
intervals. The most intervals stored was 253 for one data set with
$n=7776$. These results suggest that our proposed GPDPA only stores on
average $O(\log n)$ intervals (possible changepoints), as in the
original PDPA. The overall empirical time complexity is thus
$O(K n \log n)$ for $K$ segments and $n$ data points.

We recorded the timings of each algorithm for computing models with up
to $K=19$ segments (a total of 10 peak models $k\in\{1,3,\dots,19\}$,
from 0 to 9 peaks). Since $K$ is constant, the expected time
complexity was $O(n^2)$ for the CDPA and $O(n \log n)$ for the PDPA
and GPDPA. In agreement with these expectations, our proposed GPDPA
shows $O(n\log n)$ asymptotic timings similar to the PDPA (right of
Figure~\ref{fig:timings}). 

It is clear that the $O(n^2)$ CDPA algorithm is slower than the other
two algorithms, especially for larger data sets. For the largest count
data vector ($n=263169$), the CDPA took over two hours, but the GPDPA
took only
about two minutes. Our proposed GPDPA is nearly as fast as MACS
\citep{MACS}, a heuristic from the bioinformatics literature which
took about 1 minute to compute 10 peak models for this data set. 

The total computation time to process all 2752 count data vectors was
156 hours for the CDPA, and only 6 hours for the GPDPA (26 times
faster). Overall, these results suggest that our proposed GPDPA enjoys
$O(n\log n)$ time complexity in ChIP-seq data, which makes it possible
to use for very large data sets.

\subsection{Test accuracy in ChIP-seq data}


\begin{figure*}[t!]
  \centering
  \includegraphics[width=\textwidth]{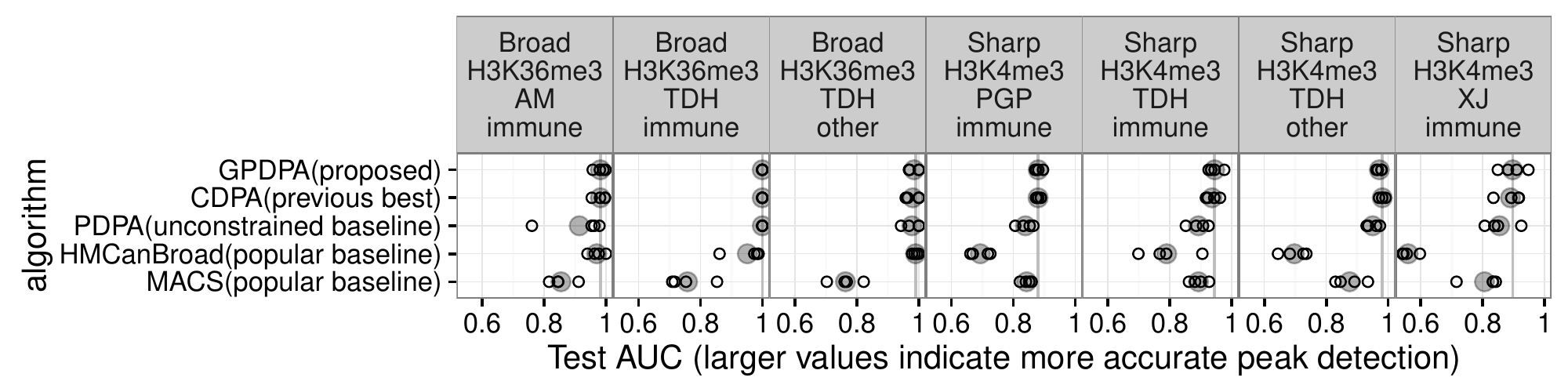}
  \vskip -0.5cm
  \caption{Four-fold cross-validation 
was used to estimate peak detection accuracy. 
    Each panel shows one of seven ChIP-seq data sets, 
    labeled by experiment (Broad H3K36me3), 
    labeler (AM), and cell types (immune).
    Each black circle shows test AUC in one of four
    cross-validation folds, the shaded grey circle is the mean, and
    the vertical line is the maximum mean in each data set. It is
    clear that the proposed GPDPA is
    just as accurate as the previous state-of-the-art CDPA, and both are
    more accurate than the other baseline methods. 
  }
  \label{fig:test-error-dots}
\end{figure*}

For the optimal changepoint detection algorithms (CDPA, PDPA, GPDPA),
the prediction problem simplifies to selecting the number of segments
$K_i\in \{1, 3,\dots, 19\}$ for each data vector $i$, resulting in a
predicted peak vector $c^{K_i}(\mathbf y_i)\in\{0,1\}^n$. We select the
number of segments using an oracle penalty
$K_i^\lambda=\argmin_k l_{ik} + \lambda o_{ik}$
\citep{cleynen2013segmentation}, where $l_{ik}$ is the Poisson loss and
$o_{ik}$ is the oracle model complexity for the model with $k$
segments for data vector $i$. 
The problem thus simplifies to learning a scalar
penalty constant $\lambda$,
\begin{equation}
  \label{eq:learn-lambda}
  \minimize_{\lambda}
  \sum_{i=1}^m E\left[
    c^{K_i^\lambda}(\mathbf y_i), 
    L_i\right].
\end{equation}


To demonstrate that changepoint detection algorithms are more accurate
than typical heuristics from the bioinformatics literature, we also
compared with the MACS and HMCanBroad methods \citep{MACS,
  HMCan}. MACS is a popular heuristic for data with a sharp peak
pattern such as H3K4me3, and \mbox{HMCanBroad} is a popular heuristic
for data with a broad peak pattern such as H3K36me3. Although they are
not designed for supervised learning, we trained them by performing
grid search over a single significance threshold parameter (qvalue for
MACS and finalThreshold for HMCanBroad).

In each of the seven data sets in the histone benchmark,
we performed four-fold cross-validation and computed test AUC (area
under the Receiver Operating Characteristic curve) to estimate the
accuracy of each algorithm. The previous algorithm with
state-of-the-art accuracy on this benchmark was the CDPA, which
enforces the up-down constraint on segment means. We expected our
proposed GPDPA to perform just as well, since it also enforces that
constraint. In agreement with our expectation, we observed that the
CDPA and GPDPA yield comparable test AUC in all seven data sets
(Figure~\ref{fig:test-error-dots}). In contrast, the unconstrained
PDPA had much lower test AUC in several data sets, because of lower
true positive rates. These results provide convincing evidence that
the constraint is necessary for optimal peak detection accuracy.

Since the baseline HMCanBroad algorithm was designed for data with a
broad peak pattern, we expected it to perform well in the H3K36me3
data. In agreement with this expectation, HMCanBroad showed
state-of-the-art test AUC in two H3K36me3 data sets (broad peak
pattern), but was very inaccurate in four H3K4me3 data sets (sharp
peak pattern). We expected the baseline MACS algorithm to perform well
in the H3K4me3 data sets, since it was designed for data with a sharp
peak pattern. In contrast to this expectation, MACS had test AUC
values much lower than the optimization-based algorithms in all seven
data sets (Figure~\ref{fig:test-error-dots}). These results suggest
that constrained optimal changepoint detection algorithms are more
accurate than the heuristics from the bioinformatics literature.


\section{Discussion and conclusions}
\label{sec:discussion}

Algorithms for changepoint detection can be classified in terms of
time complexity, optimality, constraints, and pruning techniques
(Table~1). In this paper, we investigated generalizing the functional
pruning technique originally discovered by \citet{pruned-dp} and
\citet{phd-johnson}. We showed that the functional pruning technique can
be used to compute optimal changepoints subject to affine constraints
on adjacent segment mean parameters.

We showed that our proposed Generalized Pruned Dynamic Programming
Algorithm (GPDPA) enjoys the same log-linear $O(Kn\log n)$ time
complexity as the original unconstrained PDPA, when applied to peak
detection in ChIP-seq data sets (Figure~\ref{fig:timings}). However,
we observed that the up-down constrained GPDPA is much more accurate
than the unconstrained PDPA (Figure~\ref{fig:test-error-dots}). These
results suggest that the up-down constraint is necessary for computing
a changepoint model with optimal peak detection accuracy. Indeed, we
observed that the GPDPA enjoys the same state-of-the-art accuracy as
the previous best, the relatively slow quadratic $O(Kn^2)$ time
CDPA.

We observed that the heuristic algorithms which are popular in the
bioinformatics literature (MACS, HMCanBroad) are much less accurate
than the optimal changepoint detection algorithms (CDPA, PDPA,
GPDPA). In the past these sub-optimal heuristics have been preferred
because of their speed. For example, the CDPA took 2 hours to compute
10 peak models in the largest data set in the ChIP-seq benchmark,
whereas the GPDPA took 2 minutes, and the MACS heuristic took 1
minute. Using our proposed GPDPA, it is now possible to compute highly
accurate models in an amount of time that is comparable to heuristic
algorithms. Our proposed GPDPA can now be used as an optimal
alternative to heuristic algorithms, even for large data sets.

For future work we will be interested in exploring pruning techniques
for other constrained changepoint models. When the number of expected
changepoints grows with the number of data points, then $K=O(n)$ and
our proposed GPDPA has $O(n^2 \log n)$ average time complexity (since
it computes all models with $1,\dots,K$ segments). We have already
started modifying the GPDPA for optimal partitioning
\citep{optimal-partitioning}, which results in the Generalized
Functional Prunining Optimal Partitioning (GFPOP) algorithm
(Section~\ref{sec:GFPOP}). It computes the $K$-segment model for a
single penalty constant $\lambda$ (without computing models with
$1,\dots,K-1$ segments) in $O(n\log n)$ time.

\section{Reproducible Research Statement}

The source code and data used to create this manuscript (including all
figures) is available at
\url{https://github.com/tdhock/PeakSegFPOP-paper}

\section{Acknowledgements}

This work was supported by a Discovery Frontiers project grant, ``The
Cancer Genome Collaboratory,'' jointly sponsored by the Natural
Sciences and Engineering Research Council (NSERC), Genome Canada (GC),
the Canadian Institutes of Health Research (CIHR) and the Canada
Foundation for Innovation (CFI).

\newpage

\appendix
The supplementary materials begin on this page.
\section{Proof of optimality of 
dynamic programming algorithm}
\label{sec:proof}
In this section we give a proof of Theorem~\ref{thm:gpdpa}.
\begin{proof}
Case 1 and 2 follow from the definition of $\FCC_{K,t}(u)$.

We now focus on case 3.
First notice that by definition of $\FCC_{K,t+1}(u)$ (i.e. the optimal segmentation) we must have
$\FCC_{K,t+1}(u) \leq \FCC_{K,t}(u) + \ell(y_t,u)$ and also
$\FCC_{K,t+1}(u) \leq \FCC_{K-1,t}(u) + \ell(y_t,u)$. Thus we have
$\FCC_{K,t+1}(u) \leq \min \{ \FCC_{K,t}(u) , \FCC_{K-1,t}(u) \} + \ell(y_{t+1},u)$.

Now let us assume,
$$\FCC_{K,t+1}(u) < \min \{ \FCC_{K,t}(u) , \FCC_{K-1,t}(u) \} + \ell(y_{t+1},u).$$
We will show that this lead to a contradiction.

We consider the optimal segmentation
$(\mathbf u, \mathbf t)\in{\mathcal I}_{t+1}^K$ which achieves the
optimum of $\FCC_{K,t+1}(u)$. We consider two possible cases:
\begin{description}
\item[Scenario 1: $t_K < t$.]  Define $\mathbf t'$ such that for all
  $i < K$, we have $t'_i = t_i$ and $t'_K = t$.  We have
  $(\mathbf u, \mathbf t')\in{\mathcal I}_{t}^K$.  We can thus
  decompose $\FCC_{K,t+1}(u)$ as

$$\FCC_{K,t+1}(u) = \sum_{k=1}^K
  h_{t'_{k-1}, t'_k}(u_k) + \ell(y_{t+1},u).$$ 

By assumption we would recover $\sum_{k=1}^K h_{t'_{k-1}, t'_k}(u_k) < \FCC_{K,t}(u)$ which is a contradiction
by definition of $\FCC_{K,t}(u)$. 

\item[Scenario 2: $t_K=t$.]  Define $\mathbf t'$ such that for all
  $i < K-1$, we have $t'_i = t_i$ and $t'_{K-1} = t$. Also define
  $\mathbf u'$ such that for all $k \leq K-1$, we have $u'_k = u_k$.
  Thus $(\mathbf u', \mathbf t')\in{\mathcal I}_{t}^{K-1}$, and can
  then decompose $\FCC_{K,t+1}(u)$ as

$$\FCC_{K,t+1}(u) = \sum_{k=1}^K
  h_{t'_{k-1}, t'_k}(u'_k) + \ell(y_{t+1},u).$$ 

By assumption we would recover $\sum_{k=1}^{K-1} h_{t'_{k-1}, t'_k}(u'_k) < \FCC_{K-1,t}(u)$ which is a contradiction
by definition of $\FCC_{K-1,t}(u)$. 
\end{description}
\end{proof}
We have thus proved that the dynamic programming update rules can be
used for computing the optimal cost functions $C_{k,t}$.

\section{Algorithm pseudocode}
\label{sec:pseudocode}
In this section we give pseudocode for our proposed Generalized Pruned
Dynamic Programming Algorithm (GPDPA), and related algorithms.
\subsection{GPDPA for reduced isotonic regression}
We begin by providing a pseudocode solver for the simplest case, the
reduced isotonic regression problem. We propose the following data
structures and sub-routines for the computation:
\begin{itemize}
\item FunctionPiece: a data structure which represents one piece of a
  $C_{k,t}(u)$ cost function (for one interval of mean values $u$). It
  has coefficients which depend on the convex loss function $\ell$
  (for the square loss it has three real-valued coefficients $a,b,c$
  which define a function $au^2 + bu + c$). It also has two
  real-valued elements for min/max mean values
  $[\underline u, \overline u]$ of this interval, meaning the function
  $C_{k,t}(u)=au^2 + bu + c$ for all
  $u\in[\underline u, \overline u]$. Finally it stores a previous
  segment endpoint $t'$ (integer) and mean $u'$ (real).
\item FunctionPieceList: an ordered list of FunctionPiece objects,
  which exactly stores a cost function $\FCC_{k,t}(u)$ for all values
  of last segment mean $u$.
\item $\text{OnePiece}(y, \underline u, \overline u)$: a sub-routine
  that initializes a FunctionPieceList with just one FunctionPiece
  $\ell(y, u)$ defined on $[\underline u, \overline u]$.
\item $\text{MinLess}(t, f)$: an algorithm that inputs a changepoint
  and a FunctionPieceList, and outputs the corresponding min-less
  operator $f^\leq$ (another FunctionPieceList), with the previous
  changepoint set to $t'=t$ for each of its pieces. This algorithm
  also needs to store the previous mean value $u'$ for each of the
  function pieces (see pseudocode below). 
\item $\text{MinOfTwo} (f_1, f_2)$: an algorithm that inputs two
  FunctionPieceList objects, and outputs another FunctionPieceList
  object which is their minimum. 
\item $\text{ArgMin}(f)$: an algorithm that inputs a FunctionPieceList
  and outputs three values: the optimal mean $u^*=\argmin_u f(u)$, the
  previous segment end $t'$ and mean $u'$.
\item $\text{FindMean}(u, f)$ an algorithm that inputs a mean value
  and a FunctionPieceList. It finds the FunctionPiece in $f$ with mean
  $u\in[\underline u, \overline u]$ contained in its interval, then
  outputs the previous segment end $t'$ and mean $u'$ stored in that
  FunctionPiece.
\end{itemize}
The above data structures and sub-routines are used in the following
pseudocode, which describes the GPDPA for solving the reduced
isotonic regression problem.
\begin{algorithm}[H]
\begin{algorithmic}[1]
\STATE Input: data set $\mathbf y\in\RR^n$, maximum number of segments $K\in\{2,\dots, n\}$.
\STATE Output: matrices of optimal segment means $U\in\RR^{K\times K}$ 
and ends $T\in\{1,\dots,n\}^{K\times K}$
\STATE Compute min $\underline y$ and max $\overline y$ of $\mathbf y$.
\label{line:min-max}
\STATE $\FCC_{1,1}\gets \text{OnePiece}(y_1, \underline y, \overline y)$
\label{line:init-1}
\STATE for data points $t$ from 2 to $n$:
\begin{ALC@g}
  \STATE $\FCC_{1,t}\gets \text{OnePiece}(y_t, \underline y, \overline y) + \FCC_{1,t-1}$
\label{line:init-t}
\end{ALC@g}
\STATE for segments $k$ from 2 to $K$: for data points $t$ from $k$ to $n$: // dynamic programming
\label{line:for-k-t}
\begin{ALC@g}
  \STATE $\text{min\_prev}\gets \text{MinLess}(t-1, \FCC_{k-1,t-1})$ // this is $\FCC_{k-1,t-1}^\leq$
  \label{line:MinLess}
    \STATE $\text{min\_new}\gets\text{min\_prev}$ if $t=k$, 
else $\text{MinOfTwo}(\text{min\_prev}, \FCC_{k, t-1})$
  \label{line:MinOfTwo}
  \STATE $\FCC_{k,t}\gets \text{min\_new} + \text{OnePiece}(y_t, \underline y, \overline y)$
  \label{line:AddNew}
\end{ALC@g}
\STATE for segments $k$ from 1 to $K$: // decoding for every model size $k$
\label{line:for-k-decoding}
\begin{ALC@g}
  \STATE $u^*,t',u'\gets \text{ArgMin}(\FCC_{k,n})$
  \label{line:ArgMin}
  \STATE $U_{k,k}\gets u^*;\, T_{k,k}\gets t'$ // store mean of segment $k$ and end of segment $k-1$
  \label{line:decode-kk}
  \STATE for segment $s$ from $k-1$ to $1$: // decoding for every segment $s<k$
  \label{line:for-s-decoding}
  \begin{ALC@g}
    \STATE if $u' < \infty$: $u^*\gets u'$ // equality constraint active, $u_s = u_{s+1}$
    \label{line:equality-constraint-active}
    \STATE $t',u'\gets\text{FindMean}(u^*, \FCC_{s,t'})$
    \label{line:FindMean}
    \STATE $U_{k,s}\gets u^*;\, T_{k,s}\gets t'$ // store mean of segment $s$ and end of segment $s-1$
    \label{line:decode-ks}
  \end{ALC@g}
\end{ALC@g}
\caption{\label{algo:GPDPA}Generalized Pruned Dynamic Programming
  Algorithm (GPDPA) for solving the reduced isotonic regression
  problem.}
\end{algorithmic}
\end{algorithm}

Algorithm~\ref{algo:GPDPA} begins by computing the min/max on
line~\ref{line:min-max}.  The main storage of the algorithm is
$\FCC_{k,t}$, which should be initialized as a $K\times n$ array of
empty FunctionPieceList objects. The computation of $\FCC_{1,t}$ for
all $t$ occurs on lines~\ref{line:init-1}--\ref{line:init-t}. 

The dynamic programming updates occur in the for loops on
lines~\ref{line:for-k-t}--\ref{line:AddNew}. Line~\ref{line:MinLess}
uses the MinLess sub-routine to compute the temporary
FunctionPieceList min\_prev (which represents the function
$\FCC_{k-1,t-1}^\leq$). Line~\ref{line:MinOfTwo} sets the temporary
FunctionPieceList min\_new to the cost of the only possible
changepoint if $t=k$; otherwise, it uses the MinOfTwo sub-routine to
compute the cost of the best changepoint for every possible mean
value. Line~\ref{line:AddNew} adds the cost of data point $t$, and
stores the resulting FunctionPieceList in $\FCC_{k,t}$.

The decoding of the optimal segment mean $U$ (a $K\times K$ array of
real numbers) and end $T$ (a $K\times K$ array of integers) variables
occurs in the for loops on
lines~\ref{line:for-k-decoding}--\ref{line:decode-ks}. For a given
model size $k$, the decoding begins on line~\ref{line:ArgMin} by using
the ArgMin sub-routine to solve $u^* = \argmin_u \FCC_{k,n}(u)$ (the
optimal values for the previous segment end $t'$ and mean $u'$ are
also returned). Now we know that $u^*$ is the optimal mean of the last
($k$-th) segment, which occurs from data point $t'+1$ to $n$. These
values are stored in $U_{k,k}$ and $T_{k,k}$
(line~\ref{line:decode-kk}). And we already know that the optimal mean
of segment $k-1$ is $u'$.  Note that the $u'=\infty$ flag means that
the equality constraint is active
(line~\ref{line:equality-constraint-active}). The decoding of the
other segments $s<k$ proceeds using the FindMean sub-routine
(line~\ref{line:FindMean}). It takes the cost $\FCC_{s,t'}$ of the
best model in $s$ segments up to data point $t'$, finds the
FunctionPiece that stores the cost of $u^*$, and returns the new
optimal values of the previous segment end $t'$ and mean $u'$. The
mean of segment $s$ is stored in $U_{k,s}$ and the end of segment
$s-1$ is stored in $T_{k,s}$ (line~\ref{line:decode-ks}).

The time complexity of Algorithm~\ref{algo:GPDPA} is $O(K n I)$ where
$I$ is the complexity of the MinLess and MinOfTwo sub-routines, which
is linear in the number of intervals (FunctionPiece objects) that are
used to represent the cost functions. There are pathological synthetic
data sets for which the number of intervals $I=O(n)$, implying a
time complexity of $O(K n^2)$. However, the average number
of intervals in real data sets is empirically $I=O(\log n)$, so the
average time complexity of Algorithm~\ref{algo:GPDPA} is
$O(K n \log n)$.

\subsection{MinLess algorithm}
\label{sec:MinLess}
The MinLess algorithm implements the min-less operator $f^\leq$
(Definition~\ref{def:min-less}), which is an essential sub-routine of
the GPDPA. The following sub-routines are used to implement the
MinLess algorithm.

\begin{itemize}
\item $\text{GetCost}(p, u)$: an algorithm that takes a FunctionPiece
  object $p$, and a mean value $u$, and computes the cost at $u$. For
  a square loss FunctionPiece $p$ with coefficients $a,b,c\in\RR$, we
  have $\text{GetCost}(p,u)=au^2+bu+c$.
\item $\text{OptimalMean}(p)$: an algorithm that takes one
  FunctionPiece object, and computes the optimal mean value. For a
  square loss FunctionPiece $p$ we have
  $\text{OptimalMean}(p)=-b/(2a)$.
\item $\text{ComputeRoots}(p, d)$: an algorithm that takes one
  FunctionPiece object, and computes the solutions to $p(u)=d$. For
  the square loss we propose to use the quadratic formula. For other
  convex losses that do not have closed form expressions for their
  roots, we propose to use Newton's root finding method. Note that for
  some constants $d$ there are no roots, and the algorithm needs to
  report that.
\item $f.\text{push\_piece}(\underline u, \overline u, p, u')$: push a
  new FunctionPiece at the end of FunctionPieceList $f$, with
  coefficients defined by FunctionPiece $p$, on interval
  $[\underline u, \overline u]$, with previous segment mean set to
  $u'$.
\item $\text{ConstPiece}(c)$: sub-routine that initializes a
  FunctionPiece $p$ with constant cost $c$ (for the square loss it
  sets $a=b=0$ in $au^2 + bu + c$).
\end{itemize}

\begin{algorithm}[H]
\begin{algorithmic}[1]
\STATE Input: The previous segment end $t_{\text{prev}}$ (an integer), 
 and $f_{\text{in}}$ (a FunctionPieceList).
\STATE Output: FunctionPieceList $f_{\text{out}}$, initialized as an empty list.
\STATE $\text{prev\_cost} \gets\infty$
\STATE $\text{new\_lower\_limit}\gets \text{LowerLimit}(f_{\text{in}}[0])$.
\STATE $i\gets 0$; // start at FunctionPiece on the left
\STATE while $i < $ Length($f_{\text{in}}$): // continue until FunctionPiece on the right
\begin{ALC@g}
  \STATE FunctionPiece $p\gets f_{\text{in}}[i]$
  \STATE if prev\_cost = $\infty$: // look for min in this interval.
  \begin{ALC@g}
    \STATE $\text{candidate\_mean}\gets \text{OptimalMean}(p)$ 
    \STATE if $\text{LowerLimit}(p)< \text{candidate\_mean} < \text{UpperLimit}(p)$:
    \begin{ALC@g}
      \STATE $\text{new\_upper\_limit}\gets \text{candidate\_mean}$ // Minimum found in this interval.
      \STATE $\text{prev\_cost}\gets \text{GetCost}(p, \text{candidate\_mean})$
      \STATE $\text{prev\_mean}\gets \text{candidate\_mean}$
    \end{ALC@g}
    \STATE else: // No minimum in this interval.
    \begin{ALC@g}
      \STATE 
      $\text{new\_upper\_limit}\gets
 \text{UpperLimit}(p)$
    \end{ALC@g}
    \STATE $f_{\text{out}}\text{.push\_piece}(\text{new\_lower\_limit},\text{new\_upper\_limit},p,\infty)$
    \STATE $\text{new\_lower\_limit}\gets\text{new\_upper\_limit}$
    \STATE $i\gets i+1$
  \end{ALC@g}
  \STATE else: // look for equality of $p$ and prev\_cost
  \begin{ALC@g}
    \STATE $(\text{small\_root},\text{large\_root})\gets\text{ComputeRoots}(p, \text{prev\_cost})$
    \STATE if $\text{LowerLimit}(p) < \text{small\_root} < \text{UpperLimit}(p)$:
    \begin{ALC@g}
      \STATE $f_{\text{out}}\text{.push\_piece}(
      \text{new\_lower\_limit}, 
      \text{small\_root}, 
      \text{ConstPiece}(\text{prev\_cost}), 
      \text{prev\_mean})$
      \STATE $\text{new\_lower\_limit}\gets \text{small\_root}$
      \STATE $\text{prev\_cost}\gets \infty$ 
    \end{ALC@g}
    \STATE else: // no equality in this interval
    \begin{ALC@g}
      \STATE $i\gets i+1$ // continue to next FunctionPiece
    \end{ALC@g}
  \end{ALC@g}
\end{ALC@g}
  \STATE if $\text{prev\_cost} < \infty$: // ending on constant piece
  \begin{ALC@g}
    \STATE $f_{\text{out}}\text{.push\_piece}(
    \text{new\_lower\_limit}, 
    \text{UpperLimit}(p), 
    \text{ConstPiece}(\text{prev\_cost}), 
    \text{prev\_mean})$
  \end{ALC@g}
\STATE Set all previous segment end $t'=t_{\text{prev}}$  for all FunctionPieces in $f_{\text{out}}$
\caption{\label{algo:minless}MinLess algorithm.}
\end{algorithmic}
\end{algorithm}

Consider Algorithm~\ref{algo:minless} which contains pseudocode for
the computation of the min-less operator. The algorithm initializes
prev\_cost (line~3), which is a state variable that is used on line~8
to decide whether the algorithm should look for a local minimum or an
intersection with a finite cost. Since prev\_cost is initially set to
$\infty$, the algorithm begins by following the convex function pieces
from left to right until finding a local minimum. If no minimum is found in
a given convex FunctionPiece (line~15), it is simply pushed on to the
end of the new FunctionPieceList (line~16). If a minimum occurs within
an interval (line~10), the cost and mean are stored (lines 11--12),
and a new convex FunctionPiece is created with upper limit ending at
that mean value (line~16). Then the algorithm starts looking for
another FunctionPiece with the same cost, by computing the smaller
root of the convex loss function (line~20). When a FunctionPiece is
found with a root in the interval (line~21), a new constant
FunctionPiece is pushed (line~22), and the algorithm resumes searching
for a minimum. At the end of the algorithm, a constant FunctionPiece
is pushed if necessary (line~28). The complexity of this algorithm is
$O(I)$ where $I$ is the number of FunctionPiece objects in
$f_\text{in}$.

The algorithm which implements the min-more operator is
analogous. Rather than searching from left to right, it searches from
right to left. Rather than using the small root (line~21), it uses the
large root.

\subsection{Implementation details}

Some implementation details that we found to be important:
\begin{description}
\item[Weights] for data sequences that contain repeats it is
  computationally advantageous to use a run-length encoding of the
  data, and a corresponding loss function. For example if the data
  sequence 5,1,1,1,0,0,5,5 is encoded as $n=4$ counts $y_t$ 5,1,0,5 with
  corresponding weights $w_t$ 1,3,2,2 then the Poisson loss function
  for mean $\mu$ is $\ell(y_t, w_t, \mu) = w_t(\mu- y_t\log \mu)$.
\item[Mean cost] The text defines $C_{k,t}$ functions as the total
  cost. However for very large data sets the cost values will be very
  large, resulting in numerical instability. To overcome this issue we
  instead implemented update rules using the mean cost.  For weights
  $W_{t}=\sum_{i=1}^t w_i$, the update rule to compute the mean cost is
$$  C_{k,t}(\mu) = \frac{1}{W_{t}} \left[\ell(y_t, \mu) + 
W_{t-1}
\min\{ C_{k,t-1}(\mu),\, C_{k-1,t-1}^\leq(\mu)  \}\right]$$
\item[Intervals in log(mean) space] For the Poisson model of
  non-negative count data $y_t\in\{0,1,2,\dots\}$ there is no possible
  mean $\mu$ value less than 0. We thus used $\log(\mu)$ values to
  implement intervals in FunctionPiece objects. For example rather
  than storing $\mu\in[0,1]$ we store $\log\mu\in[-\infty, 0]$.
\item[Root finding] For the ComputeRoots sub-routine for the Poisson
  loss, we used Newton root finding. For the larger root we solve
  $a\log\mu + b\mu + c = 0$ (linear as $\mu\rightarrow\infty$) and for
  the smaller root we solve $a x + be^x + c = 0$ ($x=\log \mu$, linear
  as $x\rightarrow -\infty$ and $\mu\rightarrow 0$). We stop the root
  finding when the cost is near zero (absolute cost value less
  than $10^{-12}$).
\item[Storage] Since the dynamic programming update rule for $C_{k,t}$
  only depends on $C_{k-1,t-1}^\geq$ and $C_{k,t-1}$, these are the
  only functions that need to be in memory, and the rest of the cost
  functions can be stored on disk (until the decoding step). We used
  the Berkeley DB Standard Template Library to store all the $C_{k,t}$
  as a vector of FunctionPieceList objects.
\end{description}

\subsection{Penalized version of reduced isotonic regression}

\citet{fpop} proposed the Functional Pruning Optimal Partitioning
(FPOP) algorithm to solve the ``penalized'' or ``optimal
partitioning'' version of the segment neighborhood problem, where the
constraint of $K$ segments is replaced by a non-negative penalty
$\lambda\in\RR_+$ on the number of changes in the objective
function. Rather than computing all models from 1 to $K$ segments (as
in the PDPA), the FPOP algorithm computes the single model with $K$
segments (without computing models from 1 to $K-1$ segments). The same
penalization idea can be applied to models with affine constraints
between adjacent segment means.  The penalized version of the reduced
isotonic regression problem (\ref{eq:reduced}) can be stated as
\begin{align}
  \label{eq:penalized_reduced_isotonic}
  \minimize_{
    \substack{
    \mathbf m\in\RR^n\\
\mathbf c\in\{0,1\}^{n-1}
}
    } &\ \ 
  \sum_{t=1}^n \ell(y_t, m_t) + \lambda \sum_{t=1}^{n-1} I(c_t \neq 0) \\
  \text{subject to\,} &\ \ c_t = 0 \Rightarrow m_t = m_{t+1}
  \nonumber\\
&\ \ c_t = 1 \Rightarrow m_t \leq m_{t+1}.
\nonumber
\nonumber
\end{align}
Note that the $c_t$ variable is a changepoint indicator.  The same
functional pruning techniques used for the GPDPA can be exploited to
create a solver for this problem. This results in the Generalized
Functional Pruning Optimal Partitioning Algorithm (GFPOP, see
Table~\ref{tab:GFPOP}).

\begin{table}
  \centering
\begin{tabular}{c|c|c}
  & Segment Neighborhood & Optimal Partitioning\\
\hline
unconstrained & PDPA & FPOP \\
\hline
constrained & GPDPA & GFPOP 
\end{tabular}
\caption{
Algorithms for solving constrained and unconstrained versions 
of the Segment Neighborhood and Optimal Partitioning problems. 
PDPA = Pruned Dynamic Programming Algorithm, 
FPOP = Functional Pruning Optimal Partitioning, 
G = Generalized (can handle affine constraints on adjacent segment means).}
\label{tab:GFPOP}
\end{table}

Let $\overline C_{\lambda,t}(u)$ be the
penalized cost of the most likely segmentation up to data point $t$,
with last segment mean $u$. The initialization for the first data
point is $\overline C_{\lambda,1}(u) = \ell(y_1, u)$. The dynamic programming update rule
for all data points $t>1$ is
\begin{equation}
  \overline C_{\lambda,t}(u) = \ell(y_t, u) + \min\{
  \overline C_{\lambda,t-1}^\leq(u) + \lambda,\, \overline C_{\lambda,t-1}(u)
  \}.
\end{equation}
The same sub-routines described in Section~\ref{sec:MinLess} can be
used to implement the algorithm below, which solves the penalized
reduced isotonic regression problem
(\ref{eq:penalized_reduced_isotonic}).
\begin{algorithm}[H]
\begin{algorithmic}[1]
\STATE Input: data set $\mathbf y\in\RR^n$, penalty constant $\lambda\geq 0$.
\STATE Output: vectors of optimal segment means $U\in\RR^{n}$ and ends $T\in\{1,\dots,n\}^{n}$
\STATE Compute min $\underline y$ and max $\overline y$ of $\mathbf y$.
\label{line:op-min-max}
\STATE $\overline C_{\lambda,1}\gets \text{OnePiece}(y_1, \underline y, \overline y)$
\STATE for data points $t$ from 2 to $n$: // dynamic programming
\label{line:for-dp-t}
\begin{ALC@g}
  \STATE $\text{min\_prev}\gets \lambda + \text{MinLess}(t-1, \overline C_{\lambda,t-1})$
  \label{line:op-MinLess}
  \STATE $\text{min\_new}\gets \text{MinOfTwo}(\text{min\_prev}, \overline C_{\lambda, t-1})$
  \label{line:op-MinOfTwo}
  \STATE $\overline C_{\lambda,t}\gets \text{min\_new} + \text{OnePiece}(y_t, \underline y, \overline y)$
  \label{line:op-AddNew}
\end{ALC@g}
\STATE $u^*,t',u'\gets \text{ArgMin}(\overline C_{\lambda,n})$ // begin decoding
\label{line:op-ArgMin}
\STATE $i\gets 1;\, U_{i}\gets u^*;\, T_{i}\gets t'$
\label{line:op-store-i}
\STATE while $t' > 0$:
\begin{ALC@g}
  \STATE if $u' < \infty$: $u^*\gets u'$
  \STATE $t',u'\gets\text{FindMean}(u^*, \overline C_{\lambda,t'})$
  \STATE $i\gets i+1;\, U_{i}\gets u^*;\, T_{i}\gets t'$
\label{line:op-i+1}
\end{ALC@g}
\caption{\label{algo:OPIR}Generalized Functional Pruning Optimal
  Partitioning (GFPOP) for penalized reduced isotonic regression}
\end{algorithmic}
\end{algorithm}

Algorithm~\ref{algo:OPIR} begins by computing the min/max
(line~\ref{line:op-min-max}). The main storage
of the algorithm is $\overline C_{\lambda, t}$, which should be
initialized as an array of $n$ empty FunctionPieceList objects. 

The dynamic programming recursion in this algorithm has only one for
loop over data points $t$ (line~\ref{line:for-dp-t}). The penalty
constant $\lambda$ is added to all of the function pieces that result
from MinLess (line~\ref{line:op-MinLess}), before computing MinOfTwo
(line~\ref{line:op-MinOfTwo}). The last step of each dynamic
programming update is to add the cost of the new data point
(line~\ref{line:op-AddNew}).

The decoding process on lines~\ref{line:op-ArgMin}--\ref{line:op-i+1}
is essentially the same as the GPDPA (Algorithm~\ref{algo:GPDPA}). The
last segment mean and second to last segment end are first stored on
line~\ref{line:op-store-i} in $(U_1,T_1)$. For each other segment $i$,
the mean and previous segment end are stored on line~\ref{line:op-i+1}
in $(U_i,T_i)$. Note that there should be space to store $(U_i,T_i)$
parameters for up to $n$ segments. However, there are usually less
than $n$ segments, and the algorithm should return a special flag for
unused parameters, for example $(U_i=\infty, T_i=-1)$.

The time complexity of Algorithm~\ref{algo:OPIR} is $O(n I)$, where
$I$ is the time complexity of the MinLess and MinOfTwo
sub-routines. As in the GPDPA, the time complexity of these
sub-routines is linear in the number of intervals (FunctionPiece
objects) that are used to represent the $\overline C_{\lambda, t}$
cost functions. Since the number of intervals in real data is
typically $I=O(\log n)$ (see Section~\ref{sec:results_time}), the
overall time complexity of Algorithm~\ref{algo:OPIR} is on average
$O(n \log n)$.

\subsection{Generalized Functional Pruning Optimal Partitioning
  Solvers}
\label{sec:GFPOP}

The GFPOP algorithm can solve problems with more general constraints
than reduced isotonic regression. Let $G=(V,E)$ be a directed graph
that represents the model constraints (for examples see
Figure~\ref{fig:state-graphs}). The vertices $V=\{1,\dots,|V|\}$ can be
represented as integers, one for every distinct state. The edges
$E=\{1,\dots,|E|\}$ is another set of integers, each of which represents
one of the possible changes between states. Each edge/change $c\in E$
has corresponding data
$(\underline v_c, \overline v_c, \lambda_c, g_c)$ which specifies a
transition from state $\underline v_c$ to state $\overline v_c$, with
a penalty of $\lambda_c\in\RR_+$, and a constraint function
$g_c:\RR\times\RR\rightarrow\RR$.

\begin{figure*}
  \centering
\parbox{1in}{
\centering \textbf{Unconstrained}\\
  \begin{tikzpicture}[->,>=stealth',shorten >=1pt,auto,node distance=3cm,
                    thick,main node/.style={circle,draw}]

  \node[main node] (1) {1};

  \path[every node/.style={font=\sffamily\small}]
    (1) edge [loop below] node {$g_0, \lambda$} (1);
\end{tikzpicture}
}
\parbox{1in}{
\centering \textbf{Reduced\\
isotonic\\
 regression}\\
  \begin{tikzpicture}[->,>=stealth',shorten >=1pt,auto,node distance=3cm,
                    thick,main node/.style={circle,draw}]

  \node[main node] (1) {1};

  \path[every node/.style={font=\sffamily\small}]
    (1) edge [loop below] node {$g_\uparrow, \lambda$} (1);
\end{tikzpicture}
}
\parbox{1.5in}{
\centering \textbf{Peak detection}
  \begin{tikzpicture}[->,>=stealth',shorten >=1pt,auto,node distance=3cm,
                    thick,main node/.style={circle,draw}]

  \node[main node] (1) {peak};
  \node[main node] (2) [below of=1] {background};

  \path[every node/.style={font=\sffamily\small}]
    (2) edge [bend left] node {$g_\uparrow, \lambda_\uparrow$} (1)
    (1) edge [bend left] node {$g_\downarrow, \lambda_\downarrow$} (2);
\end{tikzpicture}
}
\parbox{2.5in}{
\centering \textbf{Unimodal regression}
  \begin{tikzpicture}[->,>=stealth',shorten >=1pt,auto,node distance=3cm,
                    thick,main node/.style={circle,draw}]

  \node[main node] (1) {up/down};
  \node[main node] (2) [below left of=1] {up};
  \node[main node] (3) [below right of=1] {down};

  \path[every node/.style={font=\sffamily\small}]
    (2) edge [bend left] node {$g_\uparrow, \lambda_1$} (1)
    (1) edge [loop below] node {$g_\uparrow, \lambda_2$} (1)
    (1) edge [bend left] node {$g_\downarrow, \lambda_3$} (3)
    (3) edge [loop right] node {$g_\downarrow, \lambda_4$} (3);
\end{tikzpicture}
}
\caption{Examples of state graphs for four models. Nodes represent
  states and edges represent changes. Each change has a corresponding
  penalty $\lambda$, and a function $g$ that determines what types of
  changes are possible ($g_0$ any change, $g_\uparrow$ non-decreasing,
  $g_\downarrow$ non-increasing). Even if there is no edge from a node
  to itself, it is still possible to stay in the same state without
  introducing a changepoint and penalty. Note that the Unimodal
  regression state X should be interpreted as ``can change X''
  e.g. up/down means ``can change up/down.'' }
  \label{fig:state-graphs}
\end{figure*}
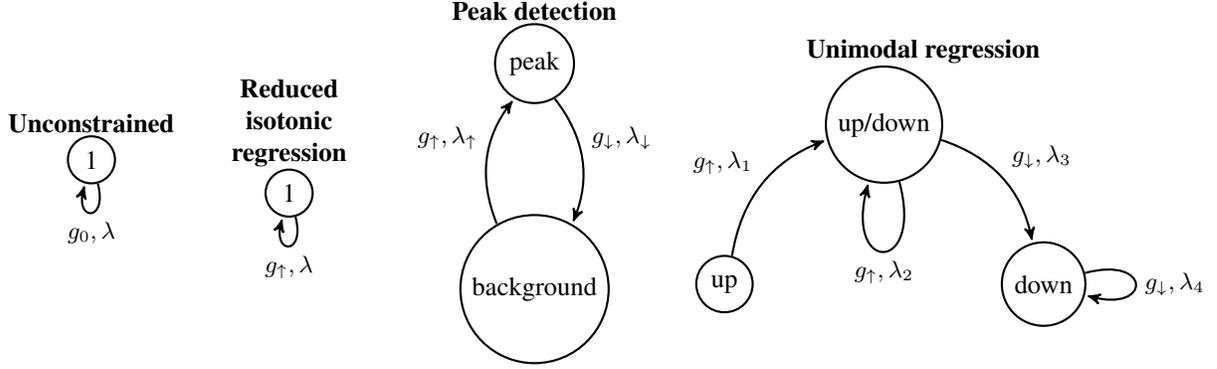
\begin{figure}[H]
  \centering
\parbox{2in}{
\centering
\textbf{Reduced Isotonic Regression}\\
  \begin{tikzpicture}[->,>=stealth',shorten >=1pt,auto,node distance=3cm,
  thick,main node/.style={circle,draw}]
  \node[main node] (t) {$C_{1,t}$};
  \node[main node] (t1) [right of=t] {$C_{1, t+1}$};
  \path[every node/.style={font=\small}]
    (t) edge [dotted] node {$C_{1, t}$} (t1)
    (t) edge [black, bend right] node [below] {$C_{1, t}^{g_\uparrow}+\lambda$} (t1)
;
\end{tikzpicture}
}
\parbox{2in}{
\centering
\textbf{Peak detection}\\
  \begin{tikzpicture}[->,>=stealth',shorten >=1pt,auto,node distance=3cm,
  thick,main node/.style={circle,draw}]
  \node[main node] (peak_t) {$C_{\text{peak},t}$};
  \node[main node] (bkg_t) [below of=peak_t] {$C_{\text{bkg}, t}$};
  \node[main node] (peak_t1) [right of=peak_t] {$C_{\text{peak}, t+1}$};
  \node[main node] (bkg_t1) [right of=bkg_t] {$C_{\text{bkg}, t+1}$};
  \path[every node/.style={font=\small}]
    (peak_t) edge [dotted] node {$C_{\text{peak}, t}$} (peak_t1)
    (peak_t) edge [black, bend right] node [right] {$C_{\text{peak}, t}^{g_\downarrow}+\lambda_\downarrow$} (bkg_t1)
    (bkg_t) edge [dotted] node[midway, below] {$C_{\text{bkg}, t}$} (bkg_t1)
    (bkg_t) edge [black, bend left] node[right] {$C_{\text{bkg}, t}^{g_\uparrow}+\lambda_\uparrow$} (peak_t1)
;
\end{tikzpicture}
}
\parbox{2in}{
\centering
\textbf{Unimodal regression}\\
  \begin{tikzpicture}[->,>=stealth',shorten >=1pt,auto,node distance=3cm,
  thick,main node/.style={circle,draw}]
  \node[main node] (up_t) {$C_{\text{up},t}$};
  \node[main node] (up_down_t) [below of=up_t] {$C_{\text{up/down}, t}$};
  \node[main node] (down_t) [below of=up_down_t] {$C_{\text{down}, t}$};
  \node[main node] (up_t1) [right of=up_t] {$C_{\text{up}, t+1}$};
  \node[main node] (up_down_t1) [below of=up_t1] {$C_{\text{up/down}, t+1}$};
  \node[main node] (down_t1) [below of=up_down_t1] {$C_{\text{down}, t+1}$};
  \path[every node/.style={font=\small}]
    (up_t) edge [dotted] node {$C_{\text{up}, t}$} (up_t1)
    (up_down_t) edge [dotted] node[midway, below] {$C_{\text{up/down}, t}$} (up_down_t1)
    (down_t) edge [dotted] node {$C_{\text{down}, t}$} (down_t1)
    (up_t) edge [black, bend left] node [near end] {$C_{\text{up}, t}^{g_\uparrow}+\lambda_{1}$} (up_down_t1)
    (up_down_t) edge [black, bend left] node [above] {$C_{\text{up/down}, t}^{g_\uparrow}+\lambda_{2}\hspace{0.8cm}$} (up_down_t1)
    (up_down_t) edge [black] node [below left] {$C_{\text{up/down}, t}^{g_\downarrow}+\lambda_3$} (down_t1)
    (down_t) edge [black, bend right] node [midway, below] {$C_{\text{down}, t}^{g_\downarrow}+\lambda_4$} (down_t1)
;
\end{tikzpicture}
}
\caption{Computation graphs which represent the dynamic programming
  updates (\ref{eq:generalDP}) for the state graph models in
  Figure~\ref{fig:state-graphs}. Nodes represent cost functions
  $C_{s,t}$ in each state $s$ at data $t$ and $t+1$. Edges represent
  inputs to the $\min\{\}$ operation (solid for a changepoint, dotted
  for no change). For example in reduced isotonic regression, $C_{1,t+1}(u) = \ell(y_{t+1}, u) + \min\{ C_{1,t}(u), C_{1,t}^{g_{\uparrow}}(u) + \lambda \}$.
}
  \label{fig:computation-graphs}
\end{figure}
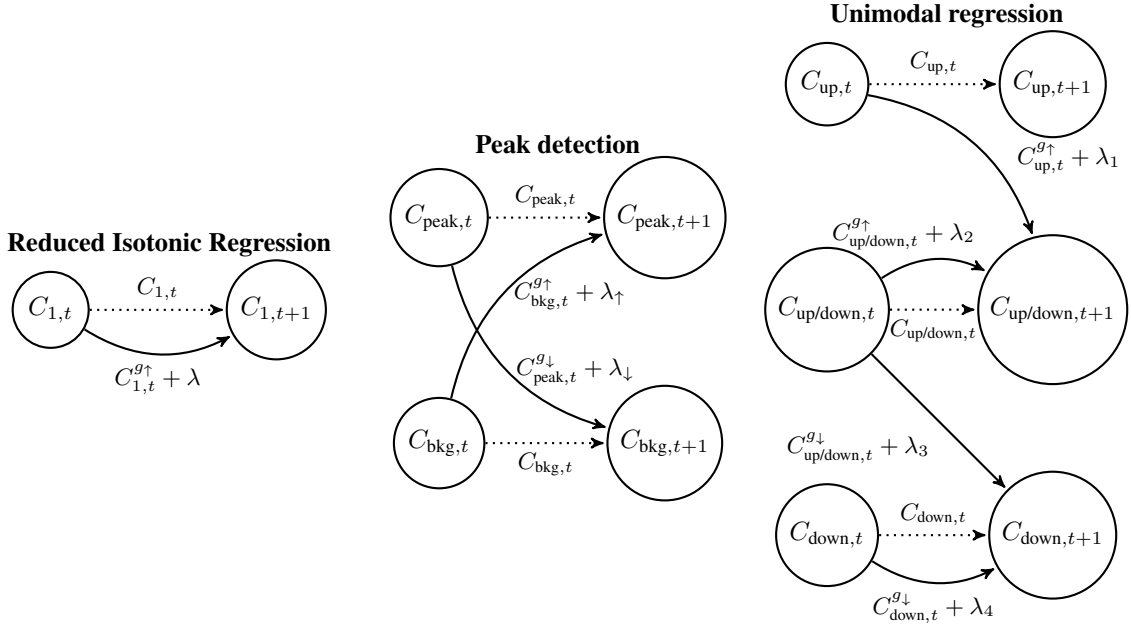

In the optimization problem below we
also allow $c=0$, which implies no penalty $\lambda_0=0$, and means no
change:
\begin{align}
  \label{eq:GFPOP_problem}
  \minimize_{
    \substack{
    \mathbf m\in\RR^n,\ \mathbf s\in V^n\\
\mathbf c\in \{0,1,\dots,|E|\}^{n-1}
}
    } &\ \ 
  \sum_{t=1}^n \ell(y_t, m_t) + \sum_{t=1}^{n-1} \lambda_{c_t} \\
  \text{subject to \hskip 0.9cm} &\ \ c_t = 0 \Rightarrow m_t = m_{t+1}
\text{ and } s_t= s_{t+1}
  \nonumber\\
&\ \ c_t \neq 0 \Rightarrow g_{c_t}(m_t, m_{t+1})\leq 0\text{ and }
(s_t,s_{t+1})=(\underline v_{c_t}, \overline v_{c_t}).
\nonumber
\end{align}
If some states are desired at the start or end, then those constraints
$s_1\in \underline S, s_n\in\overline S$ can also be enforced.  To
compute the solution to this optimization problem, we propose the
following dynamic programming algorithm.

Let $C_{s,t}(u)$ be the optimal cost with mean $u$ and state $s$ at
data point $t$. This quantity can be recursively computed using
dynamic programming. The initialization for the first data point is
$ C_{s,1}(u) = \ell(y_1, u)$ for all states $s$. The dynamic
programming update rule for all data points $t>1$ is
\begin{equation}
\label{eq:generalDP}
   C_{s,t}(u) = \ell(y_t, u) + \min\{
   M_{s,t-1}(u),\,  C_{s,t-1}(u)
  \},
\end{equation}
where the minimum cost of all possible changes to state $s$ from time
point $t-1$ is
\begin{equation}
  M_{s,t-1}(u) = \min_{c\in E_s} C^{g_c}_{\underline v_c, t-1}(u) + \lambda_c,
\end{equation}
and the set of all changes going to state $s$ is
\begin{equation}
  E_s = \{c\in E \mid \overline v_c = s\}.
\end{equation}
The computations required for the dynamic programming updates
(\ref{eq:generalDP}) can be visualized using a computation graph
(Figure~\ref{fig:computation-graphs}).

The pseudocode for the algorithm which implements the dynamic
programming updates (\ref{eq:generalDP}) is stated below.
\begin{algorithm}[H]
\begin{algorithmic}[1]
\STATE Input: data $\mathbf y$ and weights $\mathbf w$ (both size $n$), 
number of vertices/states $|V|$, starts $\underline S\subseteq V$, 
edges/transitions $E$.
\STATE Allocate $|V|\times n$ array of optimal cost functions $C_{s,t}$, 
each initialized to NULL.
\STATE for $t$ from $1$ to $n$:
\begin{ALC@g}
  \STATE if $t==1$:
  \begin{ALC@g}
    \STATE for $s$ in $\underline S$: // initialize cost for all possible starting states
    \begin{ALC@g}
      \STATE $C_{s,t}\gets\text{InitialCost}(y_t, w_t)$
    \end{ALC@g}
  \end{ALC@g}
  \STATE else:
  \begin{ALC@g}
    \STATE for $s$ from $1$ to $|V|$: 
    \begin{ALC@g}
      \STATE if $C_{s,t-1}$ is NOT NULL: // previous cost in this state has been computed
      \begin{ALC@g}
        \STATE $C_{s,t}\gets C_{s,t-1}$ // cost of staying in this state (no change)
      \end{ALC@g}
    \end{ALC@g}
    \STATE for ($\underline v$, $\overline v$, $\lambda$,
    ConstrainedCost) in $E$:
    \begin{ALC@g}
      \STATE if $C_{\underline v,t-1}$ is NOT NULL: // previous cost has been computed
      \begin{ALC@g}
        \STATE
        $\text{cost\_of\_change}\gets
        \text{ConstrainedCost}(C_{\underline v, t-1})$
        \STATE
        $\text{cost\_of\_change.set}
        (\underline v, t-1)$
        \STATE
        $\text{cost\_of\_change.addPenalty}
        (\text{$\lambda$})$
        \STATE if $C_{\overline v,t}$ is NULL:
        \begin{ALC@g}
          \STATE $C_{\overline v,t}\gets\text{cost\_of\_change}$
        \end{ALC@g}
        \STATE else:
        \begin{ALC@g}
          \STATE
          $C_{\overline v,t}\gets \text{MinOfTwo}(C_{\overline v,t},
          \text{cost\_of\_change})$
        \end{ALC@g}
      \end{ALC@g}
    \end{ALC@g}
    \STATE for $s$ from $1$ to $|V|$:
    \begin{ALC@g}
      \STATE if $C_{s,t}$ is NOT NULL:
      \begin{ALC@g}
        \STATE $C_{s,t}\text{.addDataPoint}(y_t, w_t)$
      \end{ALC@g}
    \end{ALC@g}
  \end{ALC@g}
\end{ALC@g}
\STATE Output: $|V|\times n$ array of optimal cost functions $C_{s,t}$.
\caption{\label{algo:GFPOP}Generalized Functional Pruning Optimal
  Partitioning Algorithm, Dynamic Programming (GFPOP-DP)}
\end{algorithmic}
\end{algorithm}

The algorithm above performs several checks if $C_{s,t}$ is NULL or
not (lines 9, 12, 16, 21). All costs are initialized as NULL (line
2). After having performed the cost update for data $t$, a NULL cost
$C_{s,t}$ means that state $s$ is not feasible at data $t$. For each
constraint function $g$ there is a corresponding ConstrainedCost
sub-routine that is mentioned on lines~11 and 13 (e.g. no constraint
$g_0$ MinUnconstrained, non-decreasing change $g_\uparrow$ MinLess,
non-increasing change $g_\downarrow$ MinMore). 

The average time and space complexity of Algorithm~\ref{algo:GFPOP} is
$O(|V| n I)$ where $|V|$ is the number of states and $I$ is the
average number of of intervals stored in the $|V|\times n$ array of
$C_{s,t}$ cost functions. We observed that $I=\log n$ in the empirical
tests of the peak detection model on ChIP-seq data
(Section~\ref{sec:results_time}), so we expect that the average time
complexity of Algorithm~\ref{algo:GFPOP} is $O(|V| n\log n)$.

Note that the algorithm above only performs the dynamic
programming. The decoding of optimal model parameters is achieved
using the algorithm below.

\begin{algorithm}[H]
\begin{algorithmic}[1]
\STATE Output: $|V|\times n$ array of optimal cost functions $C_{s,t}$, 
ends $\overline S\subseteq V$.
\STATE Allocate
$\mathbf m\in\RR^n$ (mean), 
$\mathbf s\in\ZZ^n$ (state), 
$\mathbf t\in\ZZ^n$ (segment end).
\STATE $u^*,s^*,t',s'\gets \text{ArgMin}(C_{\cdot,n}, \overline S)$ // begin decoding
\STATE $i\gets 1;\, m_{i}\gets u^*;\, s_i\gets s^*;\, t_{i}\gets n$
\STATE while $t' > 0$:
\begin{ALC@g}
  \STATE $i\gets i+1;\, t_{i}\gets t'$
  \STATE $u^*,s^*,t',s'\gets \text{ArgMin}(C_{s',t'})$
  \STATE $m_{i}\gets u^*;\, s_i\gets s^*$
\end{ALC@g}
\STATE Output: 
$\mathbf m$, 
$\mathbf s$, 
$\mathbf t$.
\caption{\label{algo:GFPOP-decode}Generalized Functional Pruning Optimal
  Partitioning Algorithm, decoding (GFPOP-decode)}
\end{algorithmic}
\end{algorithm}

\bibliographystyle{abbrvnat}
\bibliography{refs-abbrev}

\end{document}